\documentclass[12pt]{article}
\usepackage{amsmath,amsfonts,amssymb}
\usepackage{alltt}

\usepackage{amsmath,amsfonts,ifthen,fullpage,enumerate}  
\usepackage{mathrsfs}
\usepackage{mathtools}

\usepackage{enumerate}

\usepackage{amsmath,amsfonts,amssymb}

\usepackage{makecell}  

\usepackage{graphicx}        

\usepackage[table]{xcolor}   

\def\diam{\mathop{\rm diam}\nolimits}
\def\spam{\mathop{\rm span}\nolimits}

\def\Pol{\mathop{\rm Pol}\nolimits}
\def\Re{\mathop{\rm Re}\nolimits}

\def\pmat#1{\begin{pmatrix}#1\end{pmatrix}}

\def\question#1{{\bf Question: }#1}
\def\question#1{}

\def\cB{{\cal B}}

\def\R{\mathbb{R}}

\def\CC{\mathbb{C}}
\def\NN{\mathbb{N}}

\def\FF{\mathbb{F}}
\def\HH{\mathbb{H}}

\def\OO{\mathbb{O}}

\def\ZZ{\mathbb{Z}}

\def\Cd{\C^d}
\def\Fd{\FF^d}
\def\Od{\OO^d}
\def\Hd{{\HH^d}}

\def\Rd{\R^d}

\def\C{\mathbb{C}}

\def\SS{\mathbb{S}}

\newcommand{\RR}{\mathbb{R}}

\newtheorem{theorem}{Theorem}[section]
\newtheorem{corollary}{Corollary}[section]
\newtheorem{lemma}{Lemma}[section]
\newtheorem{example}{Example}[section]

\newtheorem{proposition}{Proposition}[section]

\newenvironment{proof}{{\noindent \it
Proof.}}{\hfill$\Box$\medskip}
\newif\ifdraft\def\draft{\drafttrue\hoffset=.8truecm\showlabeltrue
\def\comment##1{{\bf comment: ##1}}
\headline={\sevenrm \hfill \ifx\filenamed\undefined\jobname\else\filenamed\fi%
(.tex) (as of \ifx\updated\undefined???\else\updated\fi)
 \TeX'ed at {\hour\time\divide\hour by 60{}%
\minutes\hour\multiply\minutes by 60{}%
\advance\time by -\minutes
\the\hour:\ifnum\time<10{}0\fi\the\time\  on \today\hfill}}
}

\def\inpro#1{\langle#1\rangle}
\def\ip#1{\langle\kern-.28em\langle#1\rangle\kern-.28em\rangle_\nu}

\def\cH{{\cal H}}

\def\norm#1{\Vert#1\Vert}
\def\openR{{{\rm I}\kern-.16em {\rm R}}}

\def\Fd{\FF^d}
\let\ga\alpha

\let\gb\beta

\let\gG\Gamma
\let\gd\delta
\let\gD\Delta

\let\gth\theta

\let\gl\lambda

\let\gs\sigma

\let\go\omega

\let\ga\alpha

\let\gG\Gamma
\let\gb\beta
\let\gd\delta
\let\gs\sigma
\def\inpro#1{\langle#1\rangle}
\def\dinpro#1{\langle\hskip-0.25em\langle#1\rangle\hskip-0.25em\rangle}
\def\dinpro#1{\langle\hskip-0.35em\langle#1\rangle\hskip-0.35em\rangle}
\def\dinpro#1{\langle\hskip-0.35em\langle\hskip-0.35em\langle#1\rangle\hskip-0.35em\rangle\hskip-0.35em\rangle}

\def\Hom{\mathop{\rm Hom}\nolimits}

\def\Implies{\hskip1em\Longrightarrow\hskip1em}

\def\formeq{\the\sectionno.\the\equationno}  
\def\elabel#1/#2/#3/{\global\advance\equationno by 1 %
\ifx#1\empty\else\emember#1%
\ifshowlabel\marginal{\string#1}\fi\fi%
\ifmmode\eqno{#3(\formeq#2)}\else#3\formeq#2\fi} 

\def\makeblanksquare#1#2{
\dimen0=#1pt\advance\dimen0 by -#2pt
      \vrule height#1pt width#2pt depth0pt\kern-#2pt
      \vrule height#1pt width#1pt depth-\dimen0 \kern-#1pt
      \vrule height#2pt width#1pt depth0pt \kern-#2pt
      \vrule height#1pt width#2pt depth0pt
}

\title{\bf 
A variational characterisation of projective spherical 
designs over the quaternions
}

\author{Shayne Waldron\\ }

\begin{document}

\maketitle 

\begin{abstract}

We give an inequality on the packing of vectors/lines in quaternionic Hilbert space $\Hd$,
which generalises those of Sidelnikov and Welch for unit vectors in $\Rd$ and $\Cd$.
This has a parameter $t$, and depends only on the vectors up to projective
unitary equivalence.
The sequences of vectors in $\Fd=\Rd,\Cd,\Hd$ that give equality, 
which we call spherical $(t,t)$-designs, are seen to satisfy a
cubature rule on the unit sphere in $\Fd$ for a suitable polynomial space $\Hom_{\Fd}(t,t)$.
Using this, we show that the projective spherical $t$-designs 
on the Delsarte spaces $\FF P^{d-1}$ coincide with
the spherical $(t,t)$-designs of unit vectors in $\Fd$.
We then explore a number of examples in quaternionic space. 
The unitarily invariant polynomial space $\Hom_{\Hd}(t,t)$ and the inner product that we define on it so the reproducing 
kernel has a simple form are of independent interest.


\end{abstract}

\bigskip
\vfill


 \ifdraft Combinatorial aspects of finite geometries \else\fi
05B30, \ifdraft (Other designs, configurations) \else\fi
15B33, \ifdraft (Matrices over special rings (quaternions, finite fields, etc.) \else\fi
\quad
secondary
42C15, \ifdraft General harmonic expansions, frames  \else\fi
51M20, \ifdraft Polyhedra and polytopes; regular figures, division of spaces \else\fi
65D30, \ifdraft (Numerical integration) \else\fi

\vskip .5 truecm
\hrule
\newpage

\section{Introduction}

Two unit vectors $v_1$ and $v_2$ on the unit sphere in $\Fd=\Rd,\Cd$
(or the lines they represent) are spaced far apart if $|\inpro{v_1,v_2}|^2$
is {\em small}. The maximal possible separation $\inpro{v_1,v_2}=0$ occurs for
orthogonal vectors (lines). For a sequence $(v_j)$ of vectors in $\Fd$, and an
integer $t=1,2,\ldots$, the following inequality holds
\begin{equation}
\label{SidWelchineq}
\sum_{j=1}^n\sum_{k=1}^n |\inpro{v_j,v_k}|^{2t} \ge c_t(\Fd) \Bigl(
\sum_{\ell=1}^n \norm{v_\ell}^{2t}\Bigr)^2,
\end{equation}
where
\begin{equation}
\label{ctRCdefn}
c_t(\Rd)= {1\cdot3\cdot 5\cdots (2t-1)\over d(d+1)\cdots(d+2(t-1))}, \qquad
c_t(\Cd)={1\over{d+t-1\choose t}}.
\end{equation}
Sequences of vectors which give equality above can be thought of as being an
optimal packing (of well separated lines), e.g., an orthonormal basis 
gives equality for $t=1$. 
For unit vectors in $\Rd$ this inequality is due to Sidelnikov \cite{Si74},
and for unit vectors in $\Cd$ it is due to Welch \cite{W74}.
It can be shown that vectors giving equality in (\ref{SidWelchineq})
give a type of cubature rule \cite{W17}, which we call a spherical $(t,t)$-design.
From this, it follows that for any $t$ and $d$ there is equality in 
(\ref{SidWelchineq}) for some sufficiently large $n$. There is considerable interest 
in finding the smallest possible $n$ (for a given $t$ and $d$) 
\cite{HW20}, and even the case of $t$ not integral has been considered
\cite{CGGKO20}. Sequences giving this equality 
have long been used in information theory \cite{MM91}, \cite{DF07}.

The primary aim of this paper is to establish the analogue of
(\ref{SidWelchineq}) for quaternionic Hilbert space $\Hd$, and a corresponding
theory of spherical $(t,t)$-designs (cubature rules). One consequence of this, 
is that we show the projective spherical $t$-designs on the Delsarte spaces
$\FF P^{d-1}$, $\FF=\RR,\CC,\HH$ studied by Hoggar \cite{H82}, \cite{H84}
are precisely the spherical $(t,t)$-designs of unit vectors in $\Fd$. 
This gives a simple characterisation of projective spherical $t$-designs
which was previously unknown, and formally makes sense for the octonionic 
space $\OO^d$ also.

Our original intent was to extend the unified proof of (\ref{SidWelchineq})
given in \cite{W18} to quaternionic Hilbert space $\Hd$, to prove an 
inequality which was observed numerically (including the constant).
Since the quaternions $\HH$ are not commutative, this is nontrivial.
The special case $t=1$ corresponds to tight frames, and was treated in
\cite{W20}, where much of the needed theory of quaternionic Hilbert space 
was developed. Tensor products are central to the argument of \cite{W20}.
There is the ``commutativity relation''
\begin{equation}
\label{realcommute}
\Re(ab)=\Re(ba), \qquad \forall a,b\in\HH.
\end{equation}
which we hoped to apply (as in the case $t=1$) together with a 
theory of tensor products in quaternionic Hilbert space (see \cite{MW20})
to establish such an inequality.
Ultimately, this was unsuccessful, with our ``faux proof'' failing 
to clearly identify the polynomials for the cubature rule. Instead, we have adapted an 
argument of \cite{KP17} (for the complex case), which is both elegant (it uses Cauchy-Schwarz)
and insightful (equality naturally identifies the constant $c_t(\Hd)$ and the 
space $\Hom_{\Hd}(t,t)$ of polynomials for the cubature rule). 
We first present this argument for all three cases
(Theorem \ref{CSgeneralvarthm}),
and then show the existence of the inner product that it is predicated on 
(Theorems \ref{tightframeHomtttheorem} and \ref{Hdreproinproexist}).


We now give a brief summary of quaternionic Hilbert space and polynomials on it,
referring back to \cite{W20} as appropriate (also see \cite{GMP13}).

\section{Quaternionic Hilbert space}

We assume basic familiarity with 
the {\bf quaternions} $\HH$ which are an 
extension of the complex numbers $x+iy$ to a noncommutative associative algebra over the real numbers
(skew field) consisting of elements:
$$ q
=q_1+ q_2i +q_3 j +q_4 k 
=(q_1+ q_2i) +(q_3  +q_4 i)j
\in \HH, \qquad q_j\in\RR, $$
with the (noncommutative) multiplication given by
$$ i^2=j^2=k^2=-1, \quad
ij=k, \quad jk=i,\quad ki=j, \quad ji=-k,\quad kj=-i,\quad ik=-j. $$
The {\bf conjugate} and {\bf norm} of a quaternion $q=q_1+ q_2i +q_3 j +q_4 k\in\HH$ 
are given by 
$$ \overline{q}:=q_1-q_2i-q_3j-q_3k, \qquad
 |q|^2={q\overline{q}}=\overline{q}={q_1^2+q_2^2+q_3^2+q_4^2}. $$
Since the multiplication is not commutative, we must distinguish between left and 
right vector spaces (modules) over $\HH$. In \cite{W20}, we considered
$\Hd$ to be a right $\HH$-vector space (module), so that the usual rules of matrix
multiplication extend. 

We define a generalisation 
$\inpro{\cdot,\cdot}:\Hd\times\Hd\to\HH$ 
of the Euclidean inner product by  
\begin{equation}
\label{Euclideaninprodef}
\inpro{v,w}:=\sum_j \overline{v_j} w_j.
\end{equation}
Above we have used $j$ as an index, rather than as a quaternion unit, which is common practice, 
where no confusion can arise.
This {\bf inner product}
on a right $\HH$-vector space satisfies the defining conditions:

\begin{enumerate}
\item Conjugate symmetry: $\inpro{v,w}=\overline{\inpro{w,v}}$.
\item Linearity in the second variable:
$\inpro{u,v+w}=\inpro{u,v}+\inpro{u,w}$,
\\
\hbox{\hskip6.0truecm }
$\inpro{v,w\ga}=\inpro{v,w}\ga$, \\
which gives $\inpro{v\ga,w}=\overline{\ga}\inpro{v,w}$.
\item Positive definiteness: $\inpro{v,v}>0$, $v\ne 0$.
\end{enumerate}
Here linearity is in the second variable (a change of convention by
the author, to benefit from more natural formulas). Moreover, 
the {\bf Euclidean} inner product (\ref{Euclideaninprodef}) satisfies
$$ \inpro{\ga v,w}=\inpro{v,\overline{\ga}w}, \qquad\forall \ga\in\HH. $$
Much of the theory of the Euclidean inner product extends, including 
the notions of Hermitian and unitary matrices, Cauchy-Schwarz, and 
Gram-Schmidt orthogonalisation.

We now consider multivariate quaternionic polynomials $\Hd\to\HH$.
The {\bf quaternionic monomials} of degree $r$ are the
polynomials of the form
$$ q=(q_1,\ldots,q_d)\mapsto \ga_0 q_{j_1} \ga_1 q_{j_2} \ga_2 \cdots q_{j_{r-1}} \ga_{r-1} q_{j_r} \ga_r, \qquad
\ga_j\in\HH, \ j_1,\ldots, j_r \in \{1,\ldots,d\}. $$
Their $\HH$-span (as a right $\HH$-vector space) is $\Hom_r(\HH)$ the
{\bf homogeneous polynomials} of degree $r$, and $\Pol_n(\HH)$ the
{\bf polynomials} of degree $n$ is the $\HH$-span of the
homogeneous polynomials of degrees $\le n$. 


It is clear from the definitions, that the quaternionic polynomials are a graded
ring, i.e., the product of homogeneous polynomials of degrees $j$ and $k$ is
a homogeneous polynomial of degree $j+k$. To understand the dimensions of these
spaces, we write each coordinate $q_a$ of $q=(q_1,\ldots,q_d)\in\Hd$ as
$$ q_a=t_a+ix_a+jy_a+kz_a, \qquad t_a,x_a,y_a,z_a\in\RR, $$
and observe (see \cite{S79}) that
\begin{align}
\label{txyzeqns}
t_a &= {1\over4} ( q_a - iq_ai - jq_aj - kq_ak ), \qquad
x_a = {1\over4i}( q_a - iq_ai + jq_aj + kq_ak ), \cr
y_a &= {1\over4j}( q_a + iq_ai - jq_aj + kq_ak ), \qquad
z_a = {1\over4k}( q_a + iq_ai + jq_aj - kq_ak ).
\end{align}
Hence $t_a,x_a,y_a,z_a$ are homogeneous monomials (in $q_a$), as are $\overline{q_a}$ and $|q_a|^2=q_a\overline{q_a}$.

Every monomial of degree $r$ can be written as a homogeneous polynomial of
degree $r$ in the $4d$ (real) variables $t_a,x_a,y_a,z_a$, $1\le a\le d$, 
with quaternionic coefficients.
The monomials in these $4d$ real variables are linearly independent over $\HH$ by the usual
argument (of taking Taylor coefficients), and so we have
\begin{equation}
\label{Homrdim}
\dim_\HH(\Hom_r(\Hd))=\dim_\RR(\Hom_r(\RR^{4d}))= {r+4d-1\choose 4d-1}, 
\end{equation}
\begin{equation}
\label{Polnrdim}
\dim_\HH(\Pol_n(\Hd))=\dim_\RR(\Pol_n(\RR^{4d}))= {n+4d\choose 4d}.
\end{equation}
It is, at times, convenient and insightful to treat the cases $\FF=\RR,\CC,\HH$ simultaneously, 
with
\begin{equation}
\label{mdef}
m=m_\FF:=\dim_\RR(\FF) 
=\begin{cases}
1, & \FF=\RR; \cr
2, & \FF=\CC; \cr
4, & \FF=\HH.
\end{cases}
\end{equation}
throughout this paper, e.g.,
$$ \dim_\FF(\Hom_r(\Fd))=\dim_\RR(\Hom_r(\RR^{md}))= {r+md-1\choose md-1}. $$
We define a subspace of $\Hom_{2t}(\Fd)$ by
\begin{equation}
\label{Homttdefn}
\Hom_{\Fd}(t,t):=\spam\{ |\inpro{v,\cdot}|^{2t}:v\in\Fd\}, \qquad t=1,2,\ldots.
\end{equation}
Since $|\inpro{v,\cdot}|^{2t}$ maps $\Fd$ to $\RR$, we may take the span over $\RR$
or $\FF$, with the dimension unchanged. For $U:\Fd\to\Fd$ unitary, 
$$|\inpro{v,U\cdot}|^{2t}=|\inpro{U^*v,\cdot}|^{2t}, $$
so that $\Hom_{\Fd}(t,t)$ is a unitarily invariant subspace.
See Section \ref{RepHomttSect} for further detail.


\section{Integration on the real, complex and quaternionic spheres}

Though it is not immediately apparent from the inequality (\ref{SidWelchineq}) itself, 
those vectors giving equality provide discrete approximations to surface area measure on the real and complex
spheres (this is clear from Sidelnikov \cite{Si74}, but not Welch \cite{W74}).

We now provide the basic theory of integration on the sphere, 
and calculate the constant $c_t(\Hd)$, which is an average over the quaternionic sphere.
Let
$$ \SS=\SS(\Fd):=\{x\in\Fd:\norm{x}=1\} = \{x\in\RR^{md}:\norm{x}=1\} $$
be the unit sphere in $\Fd$, and 
$\gs$ be the surface area measure on $\SS$,
normalised so that $\gs(\SS)=1$.
We note that surface area measure invariant under unitary maps on $\Fd$, 
i.e., for $U$ unitary
$$ \int_{\SS(\Fd)} f(Ux)\,d\gs(x)=\int_{\SS(\Fd)} f(x)\,d\gs(x), \qquad\forall f. $$
This follows from the result for $\RR^{md}$ and the fact that the unitary maps on 
$\Fd$ correspond to a subgroup of the unitary maps on $\RR^{md}$.
Moreover, for any pair of unit vectors $x,y\in\Fd$ there is a unitary map $U$
with $y=Ux$. To prove this, take $x=e_1$ and use Gram-Schmidt and the fact that
that unitary matrices have orthonormal columns (which extend to $\Hd$).
From these observations, it follows that there is a constant $c_t(\Fd)$ with
\begin{equation}
\label{ctFdprop}
\int_{\SS(\Fd)} |\inpro{x,y}|^{2t}\,d\gs(x) = \norm{y}^{2t} c_t(\Fd), \qquad
\forall y\in\Fd.
\end{equation}
We now calculate $c_t(\Fd)$ using the well known integrals for the monomials
$$ \int_{\SS(\Rd)} x^{2\ga}\,d\gs(x) = {({1\over2})_\ga\over({d\over2})_{|\ga|}}, \qquad
\ga\in\ZZ_+^d, $$
where $x^\ga=\prod_j x_j^{\ga_j}$ and $(x)_\ga=\prod_j(x_j)_{\ga_j}$, with
$(a)_n=a(a+1)\cdots(a+n-1)$. 

\begin{lemma} 
\label{ctcalclemma}
The constant of (\ref{ctFdprop}) for $\FF=\RR,\CC,\HH$ is given by
\begin{equation}
\label{ctFdvalue}
c_t(\Fd)  
= c_{t,m}
= {({m\over2})_t\over({md\over2})_t} 
=\prod_{j=0}^{t-1}{m+2j\over md+2j}, \qquad
m:=\dim_\RR(\FF). 
\end{equation}
It satisfies $c_t(\RR)=c_t(\CC)=c_t(\HH)=1$ and
$$ c_1(\Rd)=c_1(\Cd)=c_1(\Hd)={1\over d}, \qquad
c_t(\Rd)>c_t(\Cd)>c_t(\Hd), \quad t>1,\ d>1. $$
\end{lemma}

\begin{proof} Take $x=z$, $y=e_1$ in (\ref{ctFdprop}) and
use the multinomial formula, to obtain
\begin{align*}
c_t(\Fd) 
&=\int_{\SS(\Fd)} (|z_1|^2)^{t}\,d\gs(z)
= \int_{\SS(\RR^{md})} (x_1^2+\cdots+x_m^2)^t\,d\gs(x) \cr
& = \int_{\SS(\RR^{md})} \sum_{|\ga|=t\atop\ga\in\ZZ_+^m}  {t\choose\ga} x^{2\ga}
\,d\gs(x) 
= \sum_{|\ga|=t\atop\ga\in\ZZ_+^m}  
{t\choose\ga} {({1\over2})_\ga\over({md\over2})_{|\ga|} }
= {({m\over2})_t\over({md\over2})_t}.
\end{align*}
Here the final simplification follows from  
the multivariate Rothe theorem.

The strict inequality 
$c_t(\Rd)>c_t(\Cd)$ is given in \cite{W18} (Exercise 6.9). 
We adapt the method used there for the second inequality. 
Since
$$ {c_t(\Hd)\over c_t(\Cd)} = 
{c_{t-1}(\Hd)\over c_{t-1}(\Cd)} 
{ {4+2(t-1)\over 4d+2(t-1)} \over {2+2(t-1)\over 2d+2(t-1)} } 
= {c_{t-1}(\Hd)\over c_{t-1}(\Cd)} 
\Bigl( 1 -{(t-1)(d-1)\over(t+2d-1)t}\Bigr), $$
the strict inequality holds by induction on $t$.
\end{proof}

The value (\ref{ctFdvalue})
coincides with formulas of (\ref{ctRCdefn}) for $\Rd,\Cd$, and
$$ c_t(\Hd) 
=\prod_{j=0}^{t-1}{4+2j\over 4d+2j}
= {2\cdot 3\cdots (t+1)\over 2d (2d+1)\cdots (2d+t-1)}
= {t+1\over{2d+t-1\choose t}}. $$


\section{The variational inequality}

To prove our quaternionic version of the Sidelnikov--Welch inequality (\ref{SidWelchineq}),
we require the existence 
of an inner product $\inpro{\cdot,\cdot}_\HH$ on $\Hom_{\Hd}(t,t)$ for which 
$$K_w(z):=|\inpro{w,z}|^{2t}$$
is the reproducing kernel, i.e.,
$$ \inpro{K_w,f}_\HH=f(w), \qquad\forall f\in\Hom_{\Hd}(t,t), \quad \forall w\in\Hd. $$
Such an inner product does exist (see Theorems  \ref{tightframeHomtttheorem} and \ref{Hdreproinproexist}). 
For our purposes, it is not necessary to know it explicitly 
(it follows from the reproducing property), or even the dimension of $\Hom_{\Hd}(t,t)$
(which is not obvious). 
We also take as given, the existence of such an inner product 
$\inpro{\cdot,\cdot}_\FF$ for $\Hom_{\Fd}(t,t)$, for $\FF=\RR,\CC$ also (which is well known), 
i.e.,
\begin{equation}
\label{reprokernprop}
\inpro{K_w,f}_\FF= f(w), \qquad \forall f\in\Hom_{\Fd}(t,t), \quad \forall w\in\Fd.
\end{equation}

We now prove a generalised form of (\ref{SidWelchineq})
given by Sidelnikov \cite{Si74} for $\FF=\RR$
and Kotelina and Pevnyi \cite{KP17} for $\FF=\CC$,
by using the method of the latter.

\begin{theorem}
\label{CSgeneralvarthm}
Let $\mu$ be a Borel measure on $X\subset\Fd$, $\FF=\RR,\CC,\HH$, 
$\int_X \norm{x}^{2t}\,d\mu(x)<\infty$, and $t\in\NN$.
Then
\begin{equation}
\label{WelchSidWald}
\int_X\int_X |\inpro{x,y}|^{2t}\,d\mu(x)\,d\mu(y) 
\ge c_t(\Fd) \Bigl( \int_X \norm{z}^{2t}\,d\mu(z)\Bigr)^2,
\end{equation}
with equality if and only if
\begin{equation}
\label{WelchSidWaldequalI}
{1\over\int_X \norm{x}^{2t}\, d\mu(x)}
 \int_X |\inpro{w,z}|^{2t}\,d\mu(w) = c_t(\Fd) \norm{z}^{2t}, \qquad\forall z\in\Fd,
\end{equation}
which is equivalent to the cubature rule
\begin{equation}
\label{WelchSidWaldequalcuberule}
\int_\SS f\, d\gs = {1\over\int_X \norm{x}^{2t}\, d\mu(x)} \int_X f(w)\,d\mu(w), \qquad\forall f\in\Hom_{\Fd}(t,t).
\end{equation}
There is equality in (\ref{WelchSidWald}) for $X=\SS$, $\mu=\gs$ and certain
finitely supported measures.
\end{theorem}

\begin{proof} We define polynomials $f$ and $\go_t$ in $\Hom_{\Fd}(t,t)$ by 
$$ f(z) := \int_X K_w(z)\,d\mu(w), \qquad \go_t(z):=\norm{z}^{2t}. $$
The integral defining $f(z)$ converges, since Cauchy-Schwarz gives
$$ \int_X |K_w(z)|\,d\mu(w) 
\le \int_X (\norm{w}\norm{z})^{2t}\,d\mu(w)
= \norm{z}^{2t} \int_X \norm{w}^{2t}\,d\mu(w) <\infty. $$
These are in $\Hom_{\Fd}(t,t)$, since $K_w\in\Hom_{\Fd}(t,t)$ and by
(\ref{ctFdprop}), respectively.

For the (apolar) inner product (\ref{reprokernprop}), we have
\begin{align*}
\inpro{f,\go_t}_\FF 
&= \inpro{\int_X K_w\,d\mu(w),\norm{\cdot}^{2t}}_\FF
= \int_X \inpro{ K_w ,\norm{\cdot}^{2t}}_\FF \,d\mu(w)
= \int_X \norm{w}^{2t} \,d\mu(w), \cr
\inpro{f,f}_\FF 
&= \inpro{\int_X K_x\,d\mu(x), \int_X K_y\,d\mu(y) }_\FF
= \int_X \int_X \inpro{ K_x , K_y }_\FF \,d\mu(x) \,d\mu(y) \cr
&= \int_X \int_X  |\inpro{x,y}|^{2t} \,d\mu(x) \,d\mu(y), \cr
\inpro{\go_t,\go_t}_\FF
&= \inpro{{1\over c_t(\Fd)} \int_\SS K_x\,d\gs(x),\norm{\cdot}^{2t} }_\FF
= {1\over c_t(\Fd)} \int_\SS \inpro{ K_x ,\norm{\cdot}^{2t} }_\FF \,d\gs(x) \cr
&= {1\over c_t(\Fd)} \int_\SS \norm{x}^{2t} \,d\gs(x) = {1\over c_t(\Fd)}.
\end{align*}
The integral formula for $\go_t$ used in the last equation above is (\ref{ctFdprop}). 

Thus the inequality (\ref{WelchSidWald}) is given by the Cauchy-Schwarz inequality (which also holds 
for quaternionic Hilbert space) in the form
$$ \inpro{f,f}_\FF \ge {1\over\inpro{\go_t,\go_t}_\FF } 
(\inpro{f,\go_t}_\FF)^2, $$
with equality if and only if $f$ and $\go_t$ are scalar multiples of each other, i.e.,
$$ \int_X |\inpro{w,z}|^{2t}\, d\mu(w) = C \norm{z}^{2t}. $$
The scalar $C$ above can be determined by integrating with respect to 
$\gs$, and using (\ref{ctFdprop})
\begin{align*}
 C
&= \int_\SS C \norm{z}^{2t}\,d\gs(z)
= \int_\SS \int_X |\inpro{w,z}|^{2t}\, d\mu(w)\,d\gs(z) 
= \int_X \int_\SS |\inpro{w,z}|^{2t}\, d\gs(z) \, d\mu(w) \cr
&= \int_X  c_t(\Fd) \norm{w}^{2t}\, d\mu(w)
= c_t(\Fd) \int_X  \norm{x}^{2t}\, d\mu(x). 
\end{align*}
and so we obtain the condition (\ref{WelchSidWaldequalI}) for equality.

By homogeneity, (\ref{WelchSidWaldequalI}) holds if and only if it holds
for $z\in\SS$, i.e.,
$$ {1\over\int_X \norm{x}^{2t}\, d\mu(x)}
 \int_X  K_z(w) \,d\mu(w) = c_t(\Fd) = \int_\SS K_z\, d\gs
, \qquad\forall z\in\SS(\Fd), $$
which is (\ref{WelchSidWaldequalcuberule}) for $f=K_z=|\inpro{z,\cdot}|^{2t}$.
Since the integral is linear, and $\{K_z:z\in\SS(\Fd)\}$ spans
$\Hom_{\Fd}(t,t)$, we obtain the equivalent condition (\ref{WelchSidWaldequalcuberule}).

It is easy to verify that there is equality in 
(\ref{WelchSidWald}) for $X=\SS$, $\mu=\gs$, by using (\ref{ctFdprop}),
or to observe that (\ref{WelchSidWaldequalcuberule}) holds trivially.
It follows from a result of \cite{SZ84} that (\ref{WelchSidWaldequalcuberule})
holds for a finitely supported measure.
\end{proof}

\section{Spherical $(t,t)$-designs}

Let $\gd_v$ be the Dirac $\gd$-measure concentrated at $v\in\Fd$.
A finitely supported measure
$$ \mu = \sum_{j=1}^n  w_j \gd_{v_j}, \qquad v_j\in\Fd,\ v_j\ne0, \quad w_j, > 0$$
will be called a {\bf spherical $(t,t)$-design} for $\Fd$ (or $\SS$) if it
gives equality in (\ref{WelchSidWald}), i.e., by (\ref{WelchSidWaldequalcuberule}),
$$ \int_\SS f\,d\gs = C \sum_j w_j f(v_j), \qquad\forall f\in\Hom_{\Fd}(t,t), $$
for a fixed constant $C$. 
Since $\Hom_{\Fd}(t,t)\subset\Hom_{2t}(\Fd)$,
we have
$$  C \sum_j w_j f(v_j) =  C \sum_j f\bigl( (w_j)^{1\over 2t}v_j\bigr)
= C \sum_j w_j \norm{v_j}^{2t} f\bigl({v_j\over\norm{v_j}}\bigr), \qquad
\forall f\in\Hom_{\Fd}(t,t), $$
so that the measure $\mu = \sum_j  w_j \gd_{v_j}$ giving a spherical 
$(t,t)$-design could be replaced by one where the weights $w_j$ are $1$,
or the vectors $v_j$ have unit length, i.e.,
$$ \sum_j \gd_{(w_j)^{1\over 2t}v_j}, \qquad
\sum_j \norm{v_j}^{2t} \gd_{{v_j\over\norm{v_j}}}. $$
The particular choice taken (there are many others) makes not essential
difference to the theory of spherical $(t,t)$-designs,
and we consider all such $(t,t)$-designs as equivalent.
There is some crossover with the theory of ``Euclidean $t$-designs'', 
which, in addition, seek to integrate polynomials of lower degree, and some of these 
measures (that we consider equivalent) may correspond to Euclidean designs
(see \cite{HW20}).
 
Sometimes, it is convenient for us to 
 ``normalise'' by choosing the weights to be $1$, or the vectors to be in $\SS$.
We now give the corresponding presentations of Theorem \ref{CSgeneralvarthm}.

\begin{corollary}
\label{SidWelchWalIcorollary}
Fix $t\in\NN$. Let $v_1,\ldots,v_n$ be vectors in $\Fd$, 
$\FF=\RR,\CC,\HH$, not all zero. Then
\begin{equation}
\label{SidWelchWalineqexI}
\sum_{j=1}^n\sum_{k=1}^n |\inpro{v_j,v_k}|^{2t} \ge c_t(\Fd) \Bigl(
\sum_{\ell=1}^n \norm{v_\ell}^{2t}\Bigr)^2,
\end{equation}
with equality when one of the following equivalent conditions hold
\begin{enumerate}[\rm(a)]
\item The generalised Bessel identity
\begin{equation}
\label{SidWelchWalIBes}
c_t(\Fd) \norm{x}^{2t} = {1\over\sum_{\ell=1}^n \norm{v_\ell}^{2t}} \sum_{j=1}^n|\inpro{v_j,x}|^{2t},
\qquad \forall x\in\Fd. 
\end{equation}
\item The cubature rule for $\Hom_{\Fd}(t,t)$
\begin{equation}
\label{SidWelchWalIcubature}
\int_\SS f\,d\gs = {1\over\sum_{\ell=1}^n\norm{v_\ell}^{2t}}   \sum_{j=1}^n f(v_j), \qquad
\forall f\in \Hom_{\Fd}(t,t).
\end{equation}
\end{enumerate}
\end{corollary}

\begin{proof}
Take $\mu=\sum_j \gd_{v_j}$ in Theorem \ref{CSgeneralvarthm}.
\end{proof}

In light of the above, we will say that $(v_j)\subset\Fd$ is a 
{\bf spherical $(t,t)$-design} if
\begin{equation}
\label{vjvarcharofttdesigns}
\sum_{j=1}^n\sum_{k=1}^n |\inpro{v_j,v_k}|^{2t} 
= c_t(\Fd) \Bigl( \sum_{\ell=1}^n \norm{v_\ell}^{2t}\Bigr)^2.
\end{equation}


\begin{corollary}
\label{SidWelchWalIIcorollary}
Fix $t\in\NN$. Let $v_1,\ldots,v_n$ be unit vectors in $\Fd$,
$\FF=\RR,\CC,\HH$, and $(w_j)$ be nonnegative weights with $\sum_j w_j=1$. Then
\begin{equation}
\label{SidWelchWalineqexII}
\sum_{j=1}^n\sum_{k=1}^n w_jw_k |\inpro{v_j,v_k}|^{2t} \ge c_t(\Fd),
\end{equation}
with equality when one of the following equivalent conditions hold
\begin{enumerate}[\rm(a)]
\item The generalised Bessel identity
\begin{equation}
\label{SidWelchWalIIBes}
c_t(\Fd) \norm{x}^{2t} = \sum_{j=1}^n w_j|\inpro{v_j,x}|^{2t}, \qquad \forall x\in\Fd.
\end{equation}
\item The cubature rule for $\Hom_{\Fd}(t,t)$
\begin{equation}
\label{SidWelchWalIIcubature}
\int_\SS f\,d\gs = 
  \sum_{j=1}^n w_j f(v_j), \qquad \forall f\in \Hom_{\Fd}(t,t).
\end{equation}
\end{enumerate}
\end{corollary}

\begin{proof}
Take $\mu=\sum_j w_j\gd_{v_j}$ in Theorem \ref{CSgeneralvarthm}.
\end{proof}

Since 
$|\inpro{v,\cdot}|^{2r}\norm{\cdot}^{2t-2r}\in\Hom_{\Fd}(t,t)$
(see Example \ref{zonalinHomtt}), it follows from the cubature rule characterisation
that a  spherical $(t,t)$-design is
a spherical $(r,r)$-design for $1\le r\le t$. 
This takes a more natural form in the presentation with weights $(w_j)$
and vectors on the sphere.

\begin{proposition}
\label{lowerdegdesignprop}
Let $\FF=\RR,\CC,\HH$. Then
\begin{enumerate}[\rm(a)]
\item If $(v_j)\subset\Fd$ is a spherical $(t,t)$-design (for $\Fd)$,
then $(\norm{v_j}^{t/r-1}v_j)$ is a spherical $(r,r)$-design, 
$1\le r\le t$, i.e.,
\begin{equation}
\label{lowerdegdesignpropeqI}
\sum_{j=1}^n \sum_{k=1}^n |\inpro{v_j,v_k}|^{2r}
\norm{v_j}^{2t-2r}\norm{v_k}^{2t-2r}=c_r(\Fd) \Bigl(\sum_{\ell=1}^n \norm{v_\ell}^{2t}\Bigr)^2.
\end{equation}
\item If  $(w_j)$, $(v_j)\subset\SS(\Fd)$ is a (weighted) spherical $(t,t)$-design,
then $(w_j)$, $(v_j)$ is a spherical $(r,r)$-design, $1\le r\le t$, i.e.,
\begin{equation}
\label{lowerdegdesignpropeqII}
c_r(\Fd) := \int_\SS\int_\SS |\inpro{x,y}|^{2r}\, d\gs(x)\,d\gs(y)
= \sum_{j=1}^n\sum_{k=1}^n w_jw_k |\inpro{v_j,v_k}|^{2r} .
\end{equation}
\end{enumerate}
\end{proposition}

\begin{proof}
Let $f\in\Hom_{\Fd}(r,r)$, so that $\norm{\cdot}^{2t-2r}f\in\Hom_{\Fd}(t,t)$, and
(\ref{SidWelchWalIIcubature}) gives
$$ \int_\SS f \,d\gs
= \int_\SS \norm{\cdot}^{2t-2r}f \,d\gs
= \sum_{j=1}^n w_j \norm{v_j}^{2t-2r}f(v_j)
= \sum_{j=1}^n w_j f(v_j), $$
which, by Corollary \ref{SidWelchWalIIcorollary}, gives (b).
Part (a) follows similarly from Corollary \ref{SidWelchWalIcorollary}
(also see \cite{W18} Proposition 6.2).
\end{proof}

\section{Quaternionic spherical $(t,t)$-designs}

In view of Theorem \ref{Quaterttdesigncharthm},
examples of quaternionic spherical $(t,t)$-designs
are given by the known projective $t$-designs. In particular, see 
the listing in \cite{H82}.

\begin{example} There are nine quaternionic spherical $(t,t)$-designs listed in \cite{H82}. 
All, except Examples $27$ and $30$, have rational angles $\ga=|\inpro{v_j,v_k}|^2$.
The Example 27 is a $315$ vector $(5,5)$-design for $\HH^3$,  
for which (\ref{Hoggarcdn})
holds as (see \cite{H84}, Table 3)
$$ \hbox{$1+10(0)^5+32({3-\sqrt{5}\over8})^5+160({1\over4})^5+80({1\over2})^5+32({3+\sqrt{5}\over8})^5
= {15\over2} = n\cdot c_5(\HH^3)= 315 {1\over 42}$}. $$
\end{example}


For the octonions (Cayley numbers) $\Od$, one can formally define the Euclidean inner product
as for $\RR,\CC,\HH$. It satisfies
$$ \overline{\inpro{v,w}}=\inpro{w,v}, \qquad
|\inpro{v,w}|^2=|\inpro{w,v}|^2, \qquad
\inpro{v,v}>0, \quad v\ne0, $$
but is not linear in the second variable (only additive, in both).
Therefore the variational inequality (\ref{SidWelchWalineqexI})
and {\em octonionic} spherical $(t,t)$-designs
can be defined formally.
A notion of (projective) unitary equivalence of such designs is not obvious.
It has not yet been established whether the octonionic version of the 
Welch-Sidelnikov inequality holds.

\begin{example} (MUBs) Consider the $n=2m+2$ unit vectors $(v_j)$ in $\Fd$ given by
$$ \{\pmat{1\cr0},\pmat{0\cr1}\} \cup \{ {1\over\sqrt{2}}\pmat{1\cr a}, {1\over\sqrt{2}}\pmat{1\cr -a} \}_{a\in\{1,i,j,k\}}, $$
The left-hand side of 
(\ref{vjvarcharofttdesigns}) is
$$ (2m+2)\cdot 1^t + 2m(2m+2)\cdot\Bigl({1\over2}\Bigr)^t   + (2m+2)\cdot 0^t
= 2m+2+{4\over 2^t}m(m+1), $$
and the right-hand side is
$$ c_t(\FF^2) (2m+2)^2
= { m(m+2)\cdots(m+2t-2)\over md(md+2)\cdots (md+2t-2)} (2m+2)^2. $$
These are equal for $t=1,2,3$ (and all values for $m$),
giving spherical $(3,3)$-designs.
They are Examples 1,2,3 of \cite{H82}, with Example 4 giving the octonionic version.
The ten vectors in the quaternionic case can be interpreted as a set of five mutually 
unbiased bases (or MUBs) in $\HH^2$. These meet the bound $2d+1$ on the number of 
MUBs in $\Hd$ (see \cite{CD08}). There is a general bound of ${m\over2}d+1$ on the 
number of MUBs in $\Fd$, which is obtained by this example.
\end{example}

A sequence of unit vectors $(v_j)$ in $\Fd$ (or the lines they give) is {\bf equiangular} if 
$$  |\inpro{v_j,v_k}|^2=C, \qquad j\ne k, $$
for some constant $C$. The case $C=0$ gives orthonormal vectors.

\begin{example} (SICs) It can be shown \cite{W20},
that the number of equiangular lines in $\Fd$ is 
less than or equal to $d+{m\over2}(d^2-d)$, and such a (maximal) set of
$n=d+{m\over2}(d^2-d)$ equiangular 
lines is a tight frame, with equiangularity constant $C={m\over md+2}$.
There is considerable interest in such maximal sets of equiangular lines, 
especially in the complex case, where they are known as SICs (see \cite{ACFW18}).
It follows that such a configuration is a spherical $(2,2)$-design by 
verifying
(\ref{vjvarcharofttdesigns}) via the calculation
\begin{align*} 
n +&(n^2-n)C^2 = n\bigl( 1+(n-1)C^2\bigr)= 
n \Bigl(1+\Bigl(d+{m\over2}(d^2-d)-1\Bigr)\Bigl({m\over md+2}\Bigr)^2\Bigr) \cr
&= n {(md-m+2)(m+2)\over 2(md+2)}
= n \Bigl(d+{m\over2}(d^2-d)\Bigr) {m(m+2)\over md(md+2)}
=n^2 c_2(\Fd). 
\end{align*}
There are six equiangular lines in $\HH^2$ (see \cite{K08},\cite{W20}), and 
Example 15 of \cite{H82} gives a construction of $n=2d$ equiangular lines in $\Hd$.
\end{example}

The variational characterisation (Corollary \ref{SidWelchWalIcorollary}) of spherical $(t,t)$-designs 
allows for a numerical search for them (see \cite{HW20} for the real and complex cases), 
by minimising the left-hand side of (\ref{SidWelchWalineqexI}). 
Naive calculations readily identified 
many of the known quaternionic spherical $(t,t)$-designs above 
(which have a putatively optimal number of vectors). 
We also noticed
some {\it near} designs, with rational angles (to machine precision).

\begin{example}
\label{numericalsearchI}
A numerical search for $(2,2)$-designs with a fixed number of vectors/lines in $\HH^2$, 
by minimising the left-hand side of (\ref{SidWelchWalineqexI}), gave the six equiangular lines.
A search with five vectors gave five of these six lines, 
with the variational inequality (\ref{SidWelchWalineqexI}) being
$$ 5(1)^2+20({3\over 8})^2={125\over16} =7.812500 > 7.5 = 5^2 c_2(\HH^5), $$
and a search with seven vectors
(of unit length), gave a near $(2,2)$-design, with angles ${1\over4},{1\over3},{1\over2}$ 
(to high precision), with the variational inequality being
$$ 7(1)^2+24({1\over2})^2+12({1\over3})^2+6({1\over4})^2
={353\over24}= 14.708333\cdots>14.7=c_2(\HH^2)(7)^2.  $$
\end{example}

\begin{example}
\label{numericalsearchII}
A numerical search for $(4,4)$-designs for $\HH^2$ gave various $(3,3)$-designs,
including one of $12$ vectors and one of $14$ vectors, with the corresponding
variational inequalities
$$ 12(1)^4+12(0)^4+60({2\over5})^4+60({3\over 5})^4={2664\over125}= 21.31200000
> 20.57142857\cdots = {12^2\over7}, $$
$$ 16(1)^4+80(1/5)^4+160({3\over5})^4= 36.86400000 > 36.57142857\cdots = {16^2\over7}. $$
\end{example}

Currently there is no method for determining whether or not quaternionic spherical 
$(t,t)$-designs are unitarily equivalent (as is there is in the real complex 
cases \cite{CW16}), and so it is not yet possible to see whether these numerical designs 
(and near designs) are unique up to projective unitary equivalence.

\section{Projective spherical $t$-designs on Delsarte spaces}

We now seek to make the connection between spherical $(t,t)$-designs 
(as we have defined them) and the projective spherical $t$-designs.
For this purpose, it is convenient to work with weights $(w_j)$ and 
vectors $(v_j)$ in $\SS$.

The condition of equality in (\ref{SidWelchWalineqexII}) 
used to define a spherical $(t,t)$-design can be written
\begin{equation}
c_t(\Fd) := \int_\SS \int_\SS |\inpro{x,y}|^{2t}\, d\gs(x)\,d\gs(t)
= \sum_{j=1}^n\sum_{k=1}^n w_jw_k |\inpro{v_j,v_k}|^{2t}, 
\end{equation}
or, equivalently,
\begin{equation}
\label{univariateint}
\int_\SS \int_\SS g(|\inpro{x,y}|^{2})\, d\gs(x)\,d\gs(y)
= \sum_{j=1}^n\sum_{k=1}^n w_jw_k g(|\inpro{v_j,v_k}|^{2}), 
\end{equation}
for $g=(\cdot)^t$, the univariate monomial of degree $t$.
By (\ref{lowerdegdesignpropeqII}) of Proposition \ref{lowerdegdesignprop},
(\ref{univariateint}) also holds for the univariate monomials $g=(\cdot)^r$, $1\le r\le t-1$,
and it holds trivially for the constant monomial $g=(\cdot)^0=1$. Thus

\begin{lemma} 
\label{Hoggarconnectionlemma}
Let $\mu_m$ be the Borel (probability) measure defined on $[0,1]\subset\RR$ by
\begin{equation}
\label{mumdefn}
\int_0^1 g(s)\, d\mu_m(s):= \int_\SS\int_\SS g(|\inpro{x,y}|^2)\,d\gs(x)\,d\gs(y),
\end{equation}
so that (\ref{univariateint}) can be written as
\begin{equation}
\label{univariateintII}
\int_0^1 g\, d\mu_m
= \sum_{j=1}^n\sum_{k=1}^n w_jw_k g(|\inpro{v_j,v_k}|^{2}),
\end{equation}
and let $Q_0^{(m)},Q_1^{(m)},\ldots$ the orthogonal polynomials for the measure $\mu_m$.
Then the condition for $(w_j)$, $(v_j)\subset\SS(\FF)$ to be a 
spherical $(t,t)$-design for $\Fd$ is equivalent to the following
\begin{enumerate}[\rm(a)]
\item The equation (\ref{univariateintII}) holds for the monomial $g=(\cdot)^t$.
\item The equation (\ref{univariateintII}) holds for all $g\in\Pol_t(\RR)$.
\item The equation (\ref{univariateintII}) holds for  $g=Q_1^{(m)},\ldots,Q_t^{(m)}$, i.e.,
\begin{equation}
\label{Hoggarform}
\sum_{j=1}^n\sum_{k=1}^n w_jw_k Q_\ell^{(m)}(|\inpro{v_j,v_k}|^{2})=0,  \qquad \ell=1,\ldots,t.
\end{equation}
\end{enumerate}
\end{lemma}

\begin{proof} We have already observed the conditions (a) and (b). Since $Q_0^{(m)},\ldots,Q_t^{(m)}$ is
a basis for $\Pol_t(\RR)$, and (\ref{univariateintII}) holds trivially for the constant
polynomial $Q_0^{(m)}=1$, we obtain the condition that
(\ref{univariateintII}) holds for $Q_1^{(m)},\ldots,Q_t^{(m)}$. The orthogonality 
condition gives
$$ \int_0^1 Q_\ell^{(m)} \,d\mu_m 
=  \int_0^1 Q_\ell^{(m)}Q_0^{(m)} \,d\mu_m
= 0, \qquad \ell=1,2,\ldots, $$
and therefore we obtain (c).
\end{proof}

The (induced) measure of (\ref{mumdefn}) is absolutely continuous with respect
to Lebesgue measure, and is given by 
(see \cite{H82}, Theorem 2.11) $d\mu_m(z)=W(z)\,dz$, where
\begin{equation}
\label{inducedmeasureprop}
W(z) := {\gG({md\over2})\over\gG({m\over2})\gG({m\over2}(d-1))}
z^{{m\over2}-1} (1-z)^{{m\over2}(d-1)-1}, \qquad m:=\dim_\RR(\FF).
\end{equation}
This can be checked, using the density of polynomials in $L_1(\mu_m)$, by
the calculation
\begin{align*}
\int_0^1 z^r \, W(z)\,dz
&= {\gG({md\over2})\over\gG({m\over2})\gG({m\over2}(d-1))}
\int_0^1 z^{{m\over2}+r-1} (1-z)^{{m\over2}(d-1)-1}\, dz \cr
&= {\gG({md\over2})\over\gG({m\over2})\gG({m\over2}(d-1))} 
{ \gG({m\over2}+r)\gG({m\over2}(d-1))\over \gG({md\over2}+r)}
= { ({m\over 2})_r\over({md\over 2})_r}\cr
&= c_r(\Fd) = \int_\SS\int_\SS (|\inpro{x,y}|^2)^r\, d\gs(x)\,d\gs(y).
\end{align*}
It is evident from (\ref{inducedmeasureprop}) that the orthogonal polynomials $Q_k^{(m)}$ of 
Lemma \ref{Hoggarconnectionlemma} are Jacobi polynomials (on $[0,1]$). 
Hence, we have
\begin{align*}
Q_k^{(m)}(x) &= P_k^{({m\over2}-1,{m\over2}(d-1)-1)} (1-2x)
= {({m\over2})_k\over k!} {}_2F_1(\hbox{$-k,{md\over2}-1+k,{m\over2};x$}) \cr
& = {({m\over2})_k\over k!} \sum_{j=0}^k (-1)^j {k\choose j} {({md\over2}-1+k)_j
\over ({m\over2})_j} x^j.
\end{align*}
The condition (\ref{Hoggarform}) does not depend on the particular
normalisation of the $Q_k^{(m)}$. The norm can be calculated from the orthogonality
relations for the Jacobi polynomials
$$ \int_{-1}^1 (1-z)^\ga(1+z)^\gb P_j^{(\ga,\gb)}(z) P_k^{(\ga,\gb)}(z)\,dz
={2^{\ga+\gb+1}\over 2k+\ga+\gb+1}{\gG(k+\ga+1)\gG(k+\gb+1)\over\gG(k+\ga+\gb+1)k!}\gd_{jk}. $$
The substitution
$z=1-2x$, so that $1-z=2x$, $1+z=2(1-x)$ and $dz=-2dx$, gives
$$ \int_0^1 P_j^{(\ga,\gb)}(1-2x) P_k^{(\ga,\gb)}(1-2x)\, {\gG(\ga+\gb+2) x^\ga(1-x)^\gb \over\gG(\ga+1)\gG(\gb+1)}\, dx 
= {1\over 2k+\ga+\gb+1} {(\ga+1)_{k}(\gb+1)_{k}\over(\ga+\gb+2)_{k-1} k! }
\gd_{jk} , $$
where $(x)_{-1}:=1/(x-1)$. Taking
$\ga={m\over2}-1$, $\gb={m\over2}(d-1)-1$, then gives
$$ \int_0^1 Q_j^{(m)} Q_k^{(m)} \,d\mu_m
={1\over 2k+{md\over2}-1} { ({m\over2})_{k}({m\over2}(d-1))_{k}\over({md\over2})_{k-1}}
{1\over k!} \gd_{jk}. $$

The condition (c) of Lemma \ref{Hoggarconnectionlemma}
is essentially Hoggar's definition of a $t$-design in the projective space $\FF P^{d-1}$
(a {\bf projective $t$-design}) \cite{H82},\cite{H84},\cite{H90} which is an example of a more 
general theory of $t$-designs on Delsarte spaces (which we discuss later).
There only the case with constant weights $w_j=1$ is considered, but the ``weighted'' version
of projective $t$-designs extends in the obvious fashion, see \cite{L98}.
This connection is generally understood for $\FF=\RR,\CC$ (see \cite{RS07}, \cite{W18} Theorem 6.7),
and is a new result for $\FF=\HH$.


\begin{theorem}
\label{Quaterttdesigncharthm}
The spherical $(t,t)$-designs for $\Fd$
are precisely the (projective) $t$-designs on the Delsarte spaces $\FF P^{d-1}$,
for $\FF=\RR,\CC,\HH$.
\end{theorem}

\begin{proof} 
The Neumaier construction of $t$-designs in Delsarte spaces \cite{N81}, 
which Hoggar \cite{H82},\cite{H84} used to construct projective $t$-designs,
involves the distance $d([x],[y])=\sqrt{1-|\inpro{x,y}|^2}$ between lines
given by unit vectors $x,y\in\Fd$, which are reformulated in terms of the ``angle'' 
$|\inpro{x,y}|^2=\cos^2\gth_{xy}$. It is enough to observe that the condition (1) in \cite{H84} is condition (c) of
Lemma \ref{Hoggarconnectionlemma}, where the polynomials $Q_k$ defined in (4) are multiples of
our $Q_k^{(m)}$, since
\begin{align*}
Q_k^{(m)}(x) 
&= { ({md\over2}-1+k)_k \over k! } \sum_{j=0}^k (-1)^j { k\choose j} 
{ ({m\over2}+j)_{k-j} \over ({md\over2}-1+k+j)_{k-j} } x^j,
\end{align*}
and, with ${}_j(x):=x(x-1)\cdots(x-(j-1))=(x-j+1)_j$,
\begin{align*}
Q_k(x) &:= { ({md\over2})_{2k} \over ({m\over2})_k k!} \sum_{j=0}^k (-1)^j {k\choose j}
{ {}_j(k+{m\over2}-1) \over {}_j(2k+{md\over2}-2)} x^{k-j} \cr
&= { ({md\over2})_{2k} \over ({m\over2})_k k!} \sum_{j=0}^k (-1)^{k-j} {k\choose j}
{ {}_{k-j}(k+{m\over2}-1) \over {}_{k-j}(2k+{md\over2}-2)} x^{j} \cr
&= { ({md\over2})_{2k} \over ({m\over2})_k k!} 
(-1)^{k} \sum_{j=0}^k (-1)^{j} {k\choose j}
{ ({m\over2}+j)_{k-j} \over ({md\over2}+k+j-1)_{k-j}} x^{j}.
\end{align*}
\end{proof}

Hoggar \cite{H84} considered {\bf regular schemes} $\cB$, i.e., finite sets of projective points
(unit vectors in $\Fd$) with angles $A=\{\ga_1,\ldots,\ga_s\}\subset[0,1]$,
for which the number $d_{\ga_j}$ of points making an angle $\ga_j$ with $x\in\cB$ is
independent of $x$, e.g., those given by an orbit. 

\begin{corollary}
\label{regularschemecor}
Let $\cB$ be a regular scheme of $n$ points in $\Fd$. 
Then $\cB$ is a projective $t$-design if and only if 
\begin{equation}
\label{Hoggarcdn}
1+\ga_1^r d_{\ga_1}+\cdots+\ga_s^r d_{\ga_s} = n {({m\over2})_r\over({md\over2})_r},
\end{equation}
for $r=t$. 
\end{corollary}

\begin{proof} Since $\cB=(v_j)$ is a regular scheme, 
the condition
(\ref{vjvarcharofttdesigns})  
for
being a spherical $(t,t)$-design (and hence a $t$-design) reduces to
$$ \sum_j\sum_k |\inpro{v_j,v_k}|^{2t}
= n `\sum_k |\inpro{v_j,v_1}|^{2t}
= n \bigl( 1+\ga_1^r d_{\ga_1}+\cdots+\ga_s^r d_{\ga_s} \bigr)
= c_t(\Fd) (n)^2, $$
which (after division by $n$) is (\ref{Hoggarcdn}).
\end{proof}

This illuminates and refines the Theorem 2.4 of \cite{H84}, 
which gives the condition for a regular scheme $\cB$ to be a projective $t$-design 
is that (\ref{Hoggarcdn}) holds for $r=1,\ldots,t$.

We now consider the Delsarte space construction in more detail. 
If $(X,d)$ is a metric space with finite diameter, and $\go$ a finite measure on $X$, then it 
is a {\bf Delsarte space} (with respect to $\go$) if there exist polynomials $f_{jk}$ of
degree $\le\min\{j,k\}$, for which
\begin{equation}
\label{Delsartedefcdn}
\int_X d(a,x)^{2j} d(b,x)^{2k}\, d\go(x) = f_{jk}(d(a,b)^2), \qquad\forall j,k=0,1,2,\ldots. 
\end{equation}

The metric on lines $[x]=\{\gl x:\gl\in\FF,|\gl|=1$, $x\in\Fd$, $\norm{x}=1$ in $X=\FF P^{d-1}$ is
\begin{equation}
\label{Hoggarlinemetric}
d([x],[y]) = \sqrt{1-|\inpro{x,y}|^2}, 
\end{equation}
and the measure on $X$ is given by
$$ \int_X f([x])\,d\go([x]) = \int_\SS \tilde f(x)\,\gs(x), \qquad \tilde f(x):=f([x]). $$
The condition (\ref{Delsartedefcdn}) to be a Delsarte space is that
$$ \int_X (1-|\inpro{a,x}|^2)^j(1-|\inpro{b,x}|^2)^k\, d\go([x])
=\int_\SS (1-|\inpro{a,x}|^2)^j(1-|\inpro{b,x}|^2)^k\, d\gs(x)
 = f_{jk}(1-|\inpro{a,b}|^2), $$
which is equivalent to 
\begin{equation}
\label{altDesartecdn}
\int_\SS |\inpro{a,x}|^{2j}|\inpro{b,x}|^{2k}\, d\gs(x)
 = p_{jk}(|\inpro{a,b}|^2), \qquad\forall a,b\in\SS, 
\end{equation}
for some polynomials $p_{jk}$ with degree $\le\min\{j,k\}$.
This has been proved by \cite{N81} ($\FF=\RR,\CC$) and \cite{G67} ($\FF=\HH$), so
that $\FF P^{d-1}$ is indeed a Delsarte space. 

We can give a constructive proof of (\ref{altDesartecdn}), as follows. 
As motivation, we note that Lemma \ref{ctcalclemma} gives the special case $k=0$ (a constant polynomial)
$$ \int_\SS |\inpro{a,x}|^{2j}|\inpro{b,x}|^{0}\, d\gs(x)
=\int_\SS |\inpro{a,x}|^{2j}\, d\gs(x)
= c_j(\Fd). $$
Assume, without loss of generality, that $k\le j$. By Gram-Schmidt, for $a,b\in\SS$, we have
$$ b=(b-a\inpro{a,b})+a\inpro{a,b}, \qquad (b-a\inpro{a,b})\perp a, \qquad \norm{b-a\inpro{a,b}}=\sqrt{1-|\inpro{a,b}|^2}. $$
Thus, we may choose a unitary $U$ with 
$$ U(a\inpro{a,b})=|\inpro{a,b}|e_1, \qquad U(b-a\inpro{a,b})=\sqrt{1-|\inpro{a,b}|^2}e_2, $$
so that, by the unitary invariance of surface area measure, we have
\begin{align*} \int_\SS & |\inpro{a,x} |^{2j}|\inpro{b,x}|^{2k}\, d\gs(x)
= \int_\SS |\inpro{Ua,x}|^{2j}|\inpro{Ub,x}|^{2k}\, d\gs(x) \cr
&= \int_\SS |\inpro{e_1,x}|^{2j}|\inpro{e_1|\inpro{a,b}|+\sqrt{1-|\inpro{a,b}|^2}e_2,x}|^{2k}\, d\gs(x) \cr
&= \int_\SS |x_1|^{2j}\Bigl||\inpro{a,b}|x_1+\sqrt{1-|\inpro{a,b}|^2}x_2\Bigr|^{2k}\, d\gs(x) \cr
&= \int_\SS |x_1|^{2j}| \bigl(|\inpro{a,b}|^2|x_1|^2+(1-|\inpro{a,b}|^2)|x_2|^2+2|\inpro{a,b}|\sqrt{1-|\inpro{a,b}|^2}
\Re(x_1\overline{x_2})\bigr)^{k}\, d\gs(x). \cr
\end{align*}
It is easily verified that the integral of an odd power of $\Re(x_1\overline{x_2})$ is zero
(in each of the cases $\FF=\RR,\CC,\HH$),
and so the above integral gives a polynomial of degree $k$ in $|\inpro{a,b}|^2$.

As the calculation above suggests, for the Delsarte space $X=\FF P^{d-1}$, 
it is more convenient to work with $|\inpro{x,y}|^2$, rather than the metric $d$ of
(\ref{Hoggarlinemetric}). This view point is taken in  
the unified development of Levenshtein \cite{L98}, who gives bounds for a large
class of ``codes'' $C\subset X$, with weights $m$. In addition to a metric space $(X,d)$ with a finite measure $\go$ and
weights $m$, there
is a {\bf substitution} $\gs_s$, i.e., continuous strictly monotone function $[0,\diam(X)]\to\RR$.
In this setup, a finite set $C$ with weights $m$ (which add to $|C|$) is a {\bf weighted 
$\tau$-design} (in $X$ with respect to the substitution $\gs_s(d)$) if 
\begin{equation}
\label{taudesigndef}
\int_X\int_X g(\gs_s(d(x,y)))\, d\go(x)\,d\go(y) = 
{1\over |C|^2} \sum_{x,y\in C} g(\gs_s(d(x,y)))\, m(x)m(y), 
\end{equation}
holds for all univariate polynomials $g:\RR\to\RR$ of degree $\le\tau$. 
Since this definition depends only on $\gs_s$ up to a linear change of variables, we 
may choose $\gs_s$ to have the {\bf standard} form (to be a {\bf standard substitution})
$$ \gs_s(\diam(X))=-1\le\gs_s(d)\le 1=\gs_s(0)). $$
For $\FF P^{d-1}$, \cite{L98} takes the following variant of the metric (\ref{Hoggarlinemetric})
and the standard substitution
\begin{equation}
\label{Levenshteinelinemetric}
d([x],[y]) = \sqrt{1-|\inpro{x,y}|},  \qquad \gs_s(d) = 2(1-d^2)^2-1,
\end{equation}
where
$$
\gs_s(d([x],[y])) = 2|\inpro{x,y}|^2-1=\cos(2\gth_{xy}), \quad |\inpro{x,y}|=\cos(\gth_{xy}). $$
The general form of condition (c) of Lemma \ref{Hoggarconnectionlemma}
for a weighted $\tau$-design, as defined by (\ref{taudesigndef}), 
is given in Corollary 2.14 of \cite{L98}.

A general form of the variational inequality (Theorem \ref{CSgeneralvarthm},
Corollary \ref{SidWelchWalIcorollary}) is given 
in \cite{L98} for real and complex valued functions, which includes the 
Welch and Sidelnikov inequalities, but not our quaternionic version.
This ``inequality on the mean'' of a {\bf FDNDF} ({\bf finite dimensional nonnegative definite function}
$F:X\times X\to\FF$, where $\FF=\RR,\CC$, is as follows. A function $F$ is said to be {\bf Hermitian} if
$$ \overline{F(x,y)}=F(y,x), \qquad \forall x,y\in X, $$
and moreover to be {\bf nonnegative definite} if $F|_{C\times C}$ is positive semidefinite
for all finite subsets $C\subset X$ , i.e.,
$$ \sum_{x,y\in C}  \overline{v(x)}F(x,y)v(y)\ge 0, \qquad \forall v:X\to\CC. $$
Such an $F$ is finite dimensional if it can be written
$$ F(x,y) = \sum_{j=1}^n \overline{g_j(x)} g_j(y), $$
for finitely many functions $g_j:X\to\CC$.
Important examples of FDNDF are $\inpro{x,y}$ and $|\inpro{x,y}|^2$ on $\Fd$, $\FF=\RR,\CC$. 
We note that $|\inpro{x,y}|$ is not a FDNDF, and that products of FDNDFs are FDNDF, so that
$|\inpro{x,y}|^{2t}$ is a FDNDF. 
A FDNDF $F$ is said to satisfy the {\bf inequality on the mean} if
\begin{equation}
\label{ineqofmean}
{1\over|C|^2} \sum_{x,y\in C} F(x,y) \ge \int_X\int_X F(x,y)\,d\go(x)\,d\go(y).
\end{equation}
It is shown (Corollary 3.10 \cite{L98}) that $F$ satisfies the inequality on the mean if
$\int_X F(x,y)\,d\go(y)$ does not depend on $x\in X$. For
$$ \hbox{ $\go=\gs$ on $X=\SS(\FF)$, \quad $\FF=\RR,\CC$},\qquad F(x,y)=|\inpro{x,y}|^{2t}, $$
this condition follows from (\ref{ctFdprop}),
with (\ref{ineqofmean}) becoming the (unweighted) version of the Welch and Sidelnikov inequalities,
respectively.
A theory of quaternion valued FDNDFs could be developed (cf.\ \cite{TM14}), which would
yield a corresponding inequality on the mean, giving Theorem \ref{CSgeneralvarthm}
for $\FF=\HH$ (as a particular case).
Instead, we present our original approach, which is  more constructive.





\section{Reproducing kernels and inner products on $\Hom_{\Fd}(t,t)$}
\label{RepHomttSect}

We now establish (Theorem \ref{Hdreproinproexist}) a key fact used to prove the 
variational inequality of
Theorem \ref{CSgeneralvarthm},
i.e., the existence of an inner product on 
$$ \Hom_{\Fd}(t,t):=\spam\{|\inpro{v,\cdot}|^{2t}\}, $$
with the property that
$$ \inpro{K_w,f}_\FF=f(w), \quad \forall f\in \Hom_{\Fd}(t,t), 
\qquad K_w(z):=|\inpro{z,w}|^{2t}, $$
or, in other words, there is an inner product for which $|\inpro{v,w}|^{2t}$ is
the reproducing kernel.

Reproducing kernels for real and complex Hilbert space are well studied,
and the extension to quaternionic Hilbert space follows in the obvious way \cite{TM14}.
We say that an $\FF$-Hilbert space $\cH$ ($\FF=\RR,\CC,\HH$) consisting of functions on a set $X$
is a {\bf reproducing kernel Hilbert space} if there is a ``kernel'' $K_w\in\cH$, $w\in X$,
for which 
\begin{equation}
\label{reproducingprop}
\inpro{K_w,f}=f(w), \quad \forall f\in\cH, \quad\forall w\in X.
\end{equation}
Such a kernel $K_w(z)$ can exist if only if point evaluation is a continuous linear functional.
We now present the basic structure theorem for (finite dimensional) reproducing kernel
Hilbert spaces, in terms of tight frames. A finite set $(f_j)$ in an $\FF$-Hilbert space $\cH$ 
is a {\bf normalised tight frame} (see \cite{W20}, \cite{W18}) if 
\begin{equation}
\label{tightframeprop}
f = \sum_j f_j\inpro{f_j,f}, \qquad \forall f\in\cH.
\end{equation}

\begin{proposition}
\label{reprostructureprop}
Let $(K_w)$ be the reproducing kernel for a finite dimensional $\FF$-Hilbert space, 
with normalised tight frame $(f_j)$. Then its reproducing kernel is
$$ K_w(z) = \sum_j f_j(z) \overline{f_j(w)}. $$
\end{proposition}

\begin{proof} 
Since all linear functionals on finite dimensional Hilbert spaces
are continuous, in particular the point evaluations, the Hilbert space has a reproducing kernel.
By the tight frame expansion (\ref{tightframeprop}) and the reproducing 
property (\ref{reproducingprop}), we have
$$ K_w 
= \sum_j f_j\inpro{f_j,K_w}
= \sum_j f_j\overline{\inpro{K_w,f_j}}
= \sum_j f_j \overline{f_j(w)}, $$
so that
\begin{align*}
K_w (z) &= \inpro{K_z,K_w} 
= \inpro{ \sum_k f_k \overline{f_k(z)}, \sum_j f_j \overline{f_j(w)} }
= \sum_j \Bigl(\sum_k {f_k(z)} \inpro{ f_k , f_j } \Bigr) \overline{f_j(w)} \cr
&= \sum_j f_j(z) \overline{f_j(w)}.
\end{align*}
\vskip-1truecm
\end{proof}

The desired inner product on the spaces
$\Hom_{\Rd}(t,t)$ and $\Hom_{\Cd}(t,t)$ is well known. 
We now give these, as a consequence of the multinomial theorem.
Let $i_1,\ldots,i_4\in\HH$ be given by
\begin{equation}
\label{irdefn}
i_1:=1, \qquad i_2:=i, \qquad i_3=j, \qquad i_4:=k.
\end{equation}
For a polynomial $f(x)$, in the variables $x=(x_1,\ldots,x_d)\in\Fd$,
\begin{equation}
\label{xjr-expansion}
x_j=x_{j1}i_1+x_{j2}i_2+\cdots +x_{jm}i_m\in \FF, \quad x_{j1},\ldots,x_{jm}\in\RR, \qquad 1\le j\le d,
\end{equation}
we the define the differential operator $f(D)$ by replacing $x_{jk}$ by ${\partial\over\partial x_{jk}}$,
in the usual way.
Also for $f(z)=z^\ga\overline{x}^\gb$, $z\in\Cd$, $z_j=x_j+i y_j$, we define $f(\partial)$ to be 
the differential operator obtained by replacing $z^\ga$ and $\overline{z}^\gb$ by 
$\partial^\ga$ and $\overline{\partial}$, the multivariate Wirtinger operators given by
$$ \partial_j = {\partial\over\partial z_j} =  {1\over 2}\left({\partial\over\partial x_j}
-i {\partial\over\partial y_j} \right), \qquad
\overline{\partial}_j = {\partial\over\partial \overline{z_j}} 
=  {1\over 2}\left({\partial\over\partial x_j} +i {\partial\over\partial y_j} \right). $$


\begin{example}
\label{HomttRd}
The space $\Hom_{\Rd}(t,t)$ is $\Hom_{\Rd}(2t)$, the homogeneous polynomials
of degree $2t$. Each polynomial can be written in terms of the monomial basis
$$ f(x)=\sum_{|\ga|=2t} f_\ga x^\ga, \qquad f_\ga\in\RR. $$
By the multinomial theorem,
$$ K_w(z)
=(\inpro{w,z})^{2t}
=\Bigl(\sum_j w_jz_j\Bigr)^{2t}
=\sum_{|\ga|=2t} {2t\choose\ga} w^\ga z^\ga, $$
so that
$$ \inpro{K_w,(\cdot)^\gb}_\RR
= \inpro{\sum_{|\ga|=2t} {2t\choose\ga} w^\ga (\cdot)^\ga ,(\cdot)^\gb}_\RR
= \sum_{|\ga|=2t} {2t\choose\ga} \inpro{ (\cdot)^\ga ,(\cdot)^\gb}_\RR w^\ga 
= w^\gb, \quad\forall w,\ \forall \gb $$
if and only if
${2t\choose\ga} \inpro{(\cdot)^\ga,(\cdot)^\gb}_\RR =\gd_{\ga\gb}$, 
which gives the inner product
\begin{equation}
\label{BombieriInnerprod}
\inpro{f,g}_\RR={1\over (2t)!}\sum_{|\ga|=2t} \ga! f_\ga g_\ga
={1\over (2t)!}\sum_{|\ga|=2t} {D^\ga f(0) D^\ga g(0)\over\ga!}
={1\over(2t)!} f(D)g.
\end{equation}
\end{example}

The inner product (\ref{BombieriInnerprod}) is variously known as the 
{\it Bombieri} inner product \cite{Z94} or the {\it apolar} inner product/pairing \cite{V00}. 
In this (unitarily invariant) inner product, the monomials are orthogonal, and so it is 
not a scalar multiple of the one given by integration on $\SS$ (for which $x_1^2$ and
$x_2^2$ are not orthogonal).

\begin{example}
\label{HomttCd}
The space $\Hom_{\Cd}(t,t)$ has a basis given by the monomials 
$$ m_{\ga,\gb}:z\mapsto z^\ga\overline{z}^\gb,  \qquad |\ga|=|\gb|=t, $$
so that each $f\in\Hom_{\Cd}(t,t)$ can be written uniquely
$$ f= \sum_{|\ga|=|\gb|=t} f_{\ga,\gb}\, m_{\ga,\gb}, \qquad f_{\ga,\gb}\in\CC. $$
The multinomial theorem gives
$$ K_w(z)
=(\inpro{w,z})^{t} (\inpro{z,w})^{t}
= \Bigl(\sum_j \overline{w_j}z_j \Bigr)^{t} \Bigl(\sum_k w_k\overline{z_k} \Bigr)^{t}
= \sum_{|\ga|=t} {t\choose\ga} \overline{w}^\ga z^\ga 
\sum_{|\gb|=t} {t\choose\gb} w^\gb \overline{z}^\gb, $$
so that
\begin{align*}
\inpro{K_w, m_{a,b}}_\CC
&= \inpro{\sum_{|\ga|=|\gb|=t} {t\choose\ga}  {t\choose\gb} 
\overline{w}^\ga w^\gb m_{\ga,\gb}, m_{a,b}}_\CC \cr
&= \sum_{|\ga|=|\gb|=t} {t\choose\ga}  {t\choose\gb} 
{w}^\ga \overline{w}^\gb \inpro{ m_{\ga,\gb},m_{a,b} }_\CC 
= w^a\overline{w}^b, \qquad \forall w,\ \forall a,b, 
\end{align*}
if and only if
${t\choose\ga}  {t\choose\gb} \inpro{ m_{\ga,\gb}, m_{a,b} }_\CC =\gd_{(\ga,\gb),(a,b)}$,
which gives
\begin{equation}
\label{BombieriInnerprodC}
\inpro{f,g}_\CC 
= {1\over t!^2} \sum_{|\ga|=|\gb|=t}  \ga!\gb!\, \overline{f_{\ga,\gb}}g_{\ga,\gb}
= {1\over t!^2} \tilde f (\partial) g,
\end{equation}
where $\tilde f(z):=\overline{f(\overline{z})}$, i.e,
$\tilde f=\sum_{|\ga|=|\gb|=t}
\overline{f_{\ga,\gb}}\, m_{\ga,\gb}$.
\end{example}

The inner product (\ref{BombieriInnerprodC}) can be found in \cite{KP17} and \cite{W17}.
Initially, it seemed to be impossible to find the quaternion analogue of 
(\ref{BombieriInnerprod}) and (\ref{BombieriInnerprodC}), without 
first finding a basis for $\Hom_{\Hd}(t,t)$ and an analogue of the Wirtinger calculus. 
With hindsight, we observe that
$$ f(z)=z_j \Implies \overline{f}(z):=\overline{f(z)}=\overline{z_j} \Implies 
f(D)= {\partial\over\partial x_j} -i {\partial\over\partial y_j} = 2f(\partial_j), $$
$$ f(z)=\overline{z_j} \Implies \overline{f}(z):=\overline{f(z)}={z_j} \Implies 
f(D)= {\partial\over\partial x_j} +i {\partial\over\partial y_j} = 2f(\overline{\partial}_j), $$
so that
\begin{equation}
\label{BombieriInnerprodCII}
\inpro{f,g}_\CC 
= {1\over t!^2} \tilde f (\partial) g
= {1\over t!^2 2^{2t}} \overline{f}(D) g.
\end{equation}

In Theorem \ref{Hdreproinproexist},
we give the quaternionic analogue of the (apolar) inner products
(\ref{BombieriInnerprod}) and (\ref{BombieriInnerprodCII}).

To understand $\Hom_\Hd(t,t)$, we first consider the simple example $t=1$, $d=2$.

\begin{example} 
\label{sixbasisexample}
$\Hom_{\HH^2}(1,1)$ has the following basis of six real-valued polynomials
\begin{align}
\label{sixbasis}
|q_1|^2 &= t_1^2 + x_1^2 + y_1^2 + z_1^2, \cr
|q_2|^2 &= t_2^2 + x_2^2 + y_2^2 + z_2^2, \cr
\Re(q_1\overline{q_2})   &= t_1t_2 + x_1x_2 + y_1y_2 + z_1z_2, \cr
 \Re(q_1\overline{q_2}i) &= t_1x_2 -x_1t_2 + y_1z_2 -z_1y_2, \cr
 \Re(q_1\overline{q_2}j) &= t_1y_2 -x_1z_2 -y_1t_2 + z_1x_2, \cr
 \Re(q_1\overline{q_2}k) &= t_1z_2 + x_1y_2 -y_1x_2 -z_1t_2,
\end{align}
where $q_a=t_a+x_a i +y_a j+z_a k$. To see this,
we expand $|\inpro{v,q}|^2= \inpro{v,q}\inpro{q,v}$ as
$$ |\inpro{v,q}|^2
= (\overline{v_1}q_1+\overline{v_2}q_2)
  (\overline{q_1}v_1+\overline{q_2}v_2)
= |q_1|^2|v_1|^2+|q_2|^2|v_2|^2
 + \overline{v_1}q_1\overline{q_2}v_2+\overline{v_2}q_2\overline{q_1}v_1,$$
where
$$ \overline{v_1}q_1\overline{q_2}v_2+\overline{v_2}q_2\overline{q_1}v_1
= 2\Re(\overline{v_1}q_1\overline{q_2}v_2)
= 2\Re( q_1\overline{q_2}v_2 \overline{v_1}), $$
by (\ref{realcommute}).
Write $v_2 \overline{v_1}=\ga=\ga_1+\ga_2 i+\ga_3 j+\ga_4 k$, $\ga_j\in\RR$,
and expand, to obtain
\begin{align*}
 |\inpro{v,q}|^2
&= |q_1|^2|v_1|^2+|q_2|^2|v_2|^2 + 2\Re( q_1\overline{q_2} (\ga_1+\ga_2 i+\ga_3 j+\ga_4 k)) \cr
&= |q_1|^2|v_1|^2+|q_2|^2|v_2|^2 + 2\Re( q_1\overline{q_2})\ga_1
+2\Re( q_1\overline{q_2}i)\ga_2
+2\Re( q_1\overline{q_2}j)\ga_3
+2\Re( q_1\overline{q_2}k)\ga_4,
\end{align*}
where 
$\ga_r=\Re(\overline{\ga}i_r)=\Re(v_1\overline{v_2} i_r)$ and
$(i_1,i_2,i_3,i_4)=(1,i,j,k)$.
Thus the six linearly independent polynomials in (\ref{sixbasis}) span $\Hom_{\HH^2}(1,1)$,
and hence are a basis. 
\end{example}

The above calculation generalises, to give the following.

\begin{lemma}
\label{Hom11spanexp}
The spanning polynomials $|\inpro{v,\cdot}|^2$, $v\in\Hd$, for $\Hom_\Hd(1,1)$ can be
written
\begin{equation}
\label{H11basisexp}
|\inpro{v,q}|^2 = \sum_{j=1}^d |q_j|^2|v_j|^2
+\sum_{1\le j< k\le d} \sum_{r=1}^4 2\Re(q_j\overline{q_k}i_r) \, 
\Re(v_j\overline{v_k} i_r),
\end{equation}
where $(i_1,i_2,i_3,i_4)=(1,i,j,k)$ 
and the $d+4{d\choose2}$ polynomials 
\begin{align}
\label{H11basispolys}
p_j:&\,q\mapsto |q_j|^2,   \qquad\ \qquad\, 1\le j\le d, \cr
p_{jkr}:\,&q\mapsto \Re(q_j\overline{q_k}i_r),   \qquad
1\le j<k\le d, \  1\le r\le 4,
\end{align}
are a basis for $\Hom_\Hd(1,1)$.
\end{lemma}

\begin{proof}
The expansion (\ref{H11basisexp}) follows as in 
Example \ref{sixbasisexample}. Moreover, the polynomials in 
(\ref{H11basispolys}) are linearly independent. This is easily 
seen from the formulas for them given in 
(\ref{sixbasis}), e.g., the coordinate functionals are 
given explicitly by
\begin{align}
\label{H11basisdualfunctionals}
&f\mapsto {1\over2}\partial^2_{t_j}f(0),    \qquad\quad\ \, 1\le j\le d, \cr
&f\mapsto \partial_{t_j}\partial_{(q_k)_r}f(0),   \qquad
1\le j<k\le d, \  1\le r\le 4,
\end{align}
where $q_a=t_a+x_ai+y_aj+z_ak$.
\end{proof}

The expansion (\ref{H11basisexp}) can also be written in the symmetric 
(but redundant) form
\begin{equation}
\label{H11basisexpsymmform}
|\inpro{v,q}|^2 = \sum_{j=1}^d \sum_{k=1}^d \sum_{r=1}^4 
\Re(q_j\overline{q_k}i_r) \, 
\Re(v_j\overline{v_k} i_r).
\end{equation}
Let 
\begin{equation}
\label{PandQdefn}
P=(p_j)_{1\le j\le d}\cup(\sqrt{2}p_{jkr})_{1\le j <k\le d,1\le r\le4}, \qquad
Q=(p_{jkr})_{1\le j,k\le d,1\le r\le4},
\end{equation}
be the basis and spanning set for $\Hom_{\Hd}(1,1)$ given by
the polynomials defined in (\ref{H11basispolys}).
Since the coordinates of $P$ and $Q$ are real-valued functions,
we can take monomials in $P$ and $Q$ in the usual way, i.e.,
$$ P^\ga =
\Bigl(\prod_{1\le j\le d} p_j^{\ga_j}\Bigr)
\Bigl(\prod_{1\le j <k\le d\atop1\le r\le4} (\sqrt{2}p_{jkr})^{\ga_{jkr}}\Bigr), 
\qquad
Q^\gb=
\prod_{1\le j,k\le d\atop1\le r\le4} p_{jkr}^{\gb_{jkr}}, $$
where $\ga$ and $\gb$ are multi-indices defined on the index sets of $P$ and $Q$.



\begin{theorem}
\label{tightframeHomtttheorem}
Let $P$ and $Q$ be given by (\ref{PandQdefn}).
There is a unique inner product on $\Hom_{\Hd}(t,t)$ 
for which $({t\choose\ga}^{1/2}P^\ga)_{|\ga|=t}$ and
$({t\choose\gb}^{1/2}Q^\gb)_{|\gb|=t}$ are normalised tight frames.
The reproducing kernel for this inner product is $K_w(z)=|\inpro{z,w}|^{2t}$.
\end{theorem}

\begin{proof} In \cite{W11} a notion of {\it canonical coordinates} 
for a finite spanning sequence for a real or complex vector space was developed, from which 
it follows (Theorem 4.3 of \cite{W18}) that there is a unique (canonical) inner product for
which it is a normalised tight frame. We now appeal to the quaternionic version of this
result, which holds. Briefly, for a given spanning sequence $(f_j)_{j=1}^n$ and $f$, the set of
coefficients $c=(c_j)$, for which $f=\sum_j c_j f_j$, is an affine subspace of $\FF^n$,
and hence it has a unique element $c=c(f)\in\FF^n$ which minimises $\sum_j|c_j|^2$.
The functional $f\mapsto c(f)$ is linear, and the canonical inner product between 
$f$ and $g$ is defined to be $\inpro{c(f),c(g)}$. 

We consider 
$P$ (the argument for $Q$ being the same).
We may write (\ref{H11basisexp}) as
$$ |\inpro{w,z}|^2 = \sum_j P_j(z) \overline{P_j(w)}, \qquad P=(P_j),$$
and so the multinomial theorem gives 
$$ |\inpro{w,z}|^{2t} 
= \Bigl( \sum_j P_j(z) \overline{P_j(w)} \Bigr)^t
= \sum_{|\ga|=t} {t\choose\ga} P^\ga(z) \overline{P^\ga(w)}. $$
This motivates our choice. Let $\inpro{\cdot,\cdot}_P$ be the inner product
for which $({t\choose\ga}^{1/2}P^\ga)_{|\ga|=t}$ is a normalised tight frame
for its span $\cH$.
By Proposition \ref{reprostructureprop} and the above, the
reproducing kernel for $\cH$ is
$$ K_w(z) 
= \sum_{|\ga|=t} \Bigl({t\choose\ga}^{1/2}\Bigr)^2 P^\ga(z) \overline{P^\ga(w)}
= |\inpro{w,z}|^{2t}, $$
with $\cH=\spam\{ K_w:w\in\Hd\}=\Hom_{\Hd}(t,t)$.
\end{proof}


The polynomials $\{P^\ga\}_{|\ga|=t}$ are not a basis for $\Hom_{\Hd}(t,t)$, in general.

\begin{example}
The space $\Hom_{\HH^2}(2,2)$ has dimension $20$, and there are
$21$ 
polynomials in $\{P^\ga\}_{|\ga|=2}$. 
Therefore there is one linear dependency, which is
$$ \Re(q_1\overline{q_2})^2 +\Re(q_1\overline{q_2}i)^2 +\Re(q_1\overline{q_2}j)^2 +\Re(q_1\overline{q_2}k)^2
= |q_1|^2 |q_2|^2. $$
From this, it follows that $\{P^\ga\}_{|\ga|=t}$ is not a basis for $\Hom_{\Hd}(t,t)$, 
for $t\ge2$ (and $d>1$).
\end{example}

Theorem \ref{tightframeHomtttheorem} is sufficient to prove Theorem \ref{CSgeneralvarthm}.
We now give a more constructive version (Theorem \ref{Hdreproinproexist}). 
With the notation of (\ref{irdefn}), a simple calculation gives
\begin{equation}
\label{irqirsums}
\sum_{r=1}^m i_r q i_r =  
\begin{cases}
q, & m=1; \cr
0, & m=2; \cr
-2\overline{q}, & m=4,
\end{cases}\qquad
\sum_{r=1}^m i_r q \overline{i_r} =
\begin{cases}
q, & m=1; \cr
2q, & m=2; \cr
4\Re(q), & m=4,
\end{cases}
\qquad q\in\FF, 
\end{equation}
where $m=\dim_\RR(\FF)$. Let $\gD=\sum_j\sum_r {\partial^2\over\partial x_{jr}^2}$ be
the Laplacian on functions $\Fd\to\FF$.

\begin{proposition} 
\label{linearplanewaveaction}
For $w,v\in\Fd$, we have
\begin{equation}
\label{linlinderiv}
\inpro{w,D}\inpro{\cdot,v}
=m\inpro{w,v}, \qquad
\inpro{w,D}\inpro{v,\cdot}=(2-m)\inpro{w,v},  
\end{equation}
\begin{equation}
\label{linlinderivII}
\inpro{D,w}\inpro{\cdot,v}=(2-m)\inpro{w,v}, 
\qquad
\inpro{D,w}\inpro{v,\cdot}=
\begin{cases}
m\inpro{v,w}, & m=1,2; \cr
4\Re\inpro{v,w}, & m=4.
\end{cases}
\end{equation}
\end{proposition}

\begin{proof} 
We use (\ref{irqirsums}).
Expanding gives
$$ \inpro{v,x} = \sum_j \overline{v_j} x_j = \sum_{j,r} \overline{v_j} x_{jr}i_r, \qquad
\inpro{x,v} = \sum_k \overline{x_k} v_k =\sum_{k,s} x_{ks}\overline{i_s} v_k, $$
so that (with $x$ the variable)
$$ \inpro{w,D}\inpro{x,v}
=  \sum_{j,r} \overline{w_j} i_r {\partial\over\partial x_{jr}} 
\sum_{k,s} \overline{i_s} v_k x_{ks}
=  \sum_{j,r} \overline{w_j} i_r \overline{i_r} v_j  = m\inpro{w,v}. $$
Similarly, taking the sum over $r$ for $m=4$, we have
$$ \inpro{w,D}\inpro{v,x}
=  \sum_{j,r} \overline{w_j} i_r {\partial\over\partial x_{jr}} \sum_{k,s} \overline{v_k} i_s x_{ks} 
=  \sum_{j,r} \overline{w_j} i_r \overline{v_j} i_r =  -2\sum_{j} \overline{w_j} v_j =  -2 \inpro{w,v}, $$
with the cases $m=1,2$ following by similar calculations.
This gives (\ref{linlinderiv}).

The remaining equations follow from
$$ \inpro{D,w}\inpro{x,v}
=  \sum_{j,r} \overline{i_r} w_j {\partial\over \partial x_{jr}} \sum_{k,s} \overline{i_s} v_k x_{ks}
=  \sum_{j,r} \overline{i_r} w_j \overline{i_r} v_j 
=  \sum_{j} \Bigl(\sum_r i_r w_j i_r\Bigr) v_j, 
$$
$$ \inpro{D,w}\inpro{v,x}
=  \sum_{j,r} \overline{i_r} w_j {\partial\over \partial x_{jr}} \sum_{k,s} \overline{v_k} i_s x_{ks} 
=  \sum_{j,r} \overline{i_r} w_j \overline{v_j} i_r 
=  \sum_{j} \Bigl(\sum_r i_r w_j \overline{v_j} \overline{i_r}\Bigr),
$$ 
and (\ref{irqirsums}).
\end{proof}

In view of Proposition \ref{linearplanewaveaction} the differential action of a plane wave on
a plane wave is somewhat involved in the quaternionic case. Nevertheless, we have the following.

\begin{lemma} 
\label{planewavederivlemma}
Let $\FF=\RR,\CC,\HH$.
We have the differentiation formula for plane waves
\begin{equation}
\label{planewavederiv}
|\inpro{w,D}|^2 \left(|\inpro{v,\cdot}|^{2t}\right) = 2t(2t+m-2) |\inpro{v,w}|^2 |\inpro{v,\cdot}|^{2t-2}, \qquad v,w\in\Fd,
\end{equation}
where $m=\dim_\RR(\FF)$, 
and, in particular,
\begin{equation}
\label{LaplacePlanewave}
\gD \left( |\inpro{v,\cdot}|^{2t} \right) = 2t(2t+m-2) \norm{v}^2 \inpro{v,\cdot}|^{2t-2}, \qquad v\in\Fd,
\end{equation}
\begin{equation}
\label{planewavenorm}
|\inpro{w,D}|^2 \left(\norm{\cdot}^2\right) = 2m \norm{w}^2, \qquad w\in\Fd. 
\end{equation}
\end{lemma}

\begin{proof} 
Since $|\inpro{w,D}|^2=\inpro{w,D}\inpro{D,w}$ and $|\inpro{v,\cdot}|^2$ is real-valued, the chain and product rules give
\begin{align}
\label{pwdeI}
|\inpro{w,D}|^2|\inpro{v,\cdot}|^{2t}
&= \inpro{w,D} \left( t|\inpro{v,\cdot}|^{2(t-1)} \inpro{D,w} |\inpro{v,\cdot}|^{2} \right) \cr
&=  t(t-1) |\inpro{v,\cdot}|^{2(t-2)} (\inpro{w,D}|\inpro{v,\cdot}|^{2} )(\inpro{D,w} |\inpro{v,\cdot}|^{2}) \cr
& \qquad +t|\inpro{v,\cdot}|^{2(t-1)} \inpro{w,D} \inpro{D,w} |\inpro{v,\cdot}|^{2}.
\end{align}
A calculation (to follow) gives
\begin{equation}
\label{techcalcI}
\inpro{w,D}|\inpro{v,\cdot}|^2=2\inpro{w,v}\inpro{v,\cdot},
\qquad \inpro{D,w}|\inpro{v,\cdot}|^2=2\inpro{\cdot,v}\inpro{v,w}, 
\end{equation}
\begin{equation}
\label{techcalcII}
\inpro{w,D} \inpro{D,w} |\inpro{v,\cdot}|^{2} = |\inpro{w,D}|^2  |\inpro{v,\cdot}|^{2}
=  2m |\inpro{v,w}|^2, 
\end{equation}
and so (\ref{pwdeI}) simplifies to (\ref{planewavederiv}).

Since $\gD=\norm{D}^2=|\inpro{e_1,D}|^2+\cdots+|\inpro{e_d,D}|^2$, 
we obtain (\ref{LaplacePlanewave}) from (\ref{planewavederiv}), i.e.,
\begin{align*}
\gD \left( |\inpro{v,\cdot}|^{2t}\right)
&= \sum_j |\inpro{e_j,D}|^2 |\inpro{v,\cdot}|^{2t} 
= 2t(2t+m-2) \sum_j |\inpro{v,e_j}|^2 |\inpro{v,\cdot}|^{2t-2} \cr
&= 2t(2t+m-2) \sum_j |v_j|^2 |\inpro{v,\cdot}|^{2t-2}
= 2t(2t+m-2) \norm{v}^2 |\inpro{v,\cdot}|^{2t-2}.
\end{align*}
Since $\norm{\cdot}^2=|\inpro{e_1,\cdot}|^2+\cdots|\inpro{e_d,\cdot}|^2+$,
we obtain (\ref{planewavenorm}) from (\ref{planewavederiv}), i.e.,
$$ |\inpro{w,D}|^2 \left(\norm{\cdot}^2\right) 
= \sum_j |\inpro{w,D}|^2 |\inpro{e_j,\cdot}|^2 
= \sum_j 2m |\inpro{e_j,w}|^2= 2m\sum_j|w_j|^2= 2m \norm{w}^2. $$


To prove (\ref{techcalcI}) and (\ref{techcalcII}), we need to
to use the expansion (\ref{xjr-expansion}), which gives
$$ \inpro{v,x}=\sum_{j=1}^d \overline{v_j} x_j
= \sum_{j=1}^d \sum_{r=1}^m v_{jr}\overline{i_r} \sum_{s=1}^m x_{js} i_s
= \sum_{j,r,s} \overline{i_r} i_s v_{jr} x_{js}. $$
Here the variables $v_{jr},  x_{js}\in\RR$, and so commute with any factor. Thus
\begin{align*}
\inpro{w,D} |\inpro{v,x}|^2
&=  \sum_{j,r,s} \overline{i_r} i_s w_{jr} {\partial\over\partial x_{js}}
\sum_{j_1,r_1,s_1} \overline{i_{r_1}} i_{s_1} v_{j_1r_1} x_{j_1s_1}
\sum_{j_2,r_2,s_2} \overline{i_{r_2}} i_{s_2} x_{j_2r_2} v_{j_2s_2} \cr
& =  \sum_{j,r,s} \sum_{j_1,r_1,s_1\atop j_2,r_2,s_2}
\overline{i_r} i_s \overline{i_{r_1}} i_{s_1} \overline{i_{r_2}} i_{s_2} 
w_{jr} v_{j_1r_1} v_{j_2s_2} {\partial\over\partial x_{js}} (x_{j_1s_1} x_{j_2r_2} ).
\end{align*}
By the product rule, the derivative in the expression above is $\gd_{js,j_1s_1} x_{j_2r_2}+
x_{j_1s_1} \gd_{js,j_2r_2}$. 
Using (\ref{irqirsums}), for $m=4$, we calculate the 
$(j,s)=(j_2,r_2)$ terms to be
\begin{align*} \sum_{j,r,s} \sum_{j_1,r_1,s_1\atop s_2} &
\overline{i_r} i_s \overline{i_{r_1}} i_{s_1} \overline{i_{s}} i_{s_2} 
w_{jr} v_{j_1r_1} v_{j s_2} x_{j_1s_1}
 =  \sum_{j,s,j_1} \overline{w_j} i_s \overline{v_{j_1}} x_{j_1} \overline{i_{s}} v_j 
 =  \sum_{j} \overline{w_j}\Bigl(\sum_s  i_s \inpro{v,x} \overline{i_{s}} \Bigr) v_j  \cr
& =  \sum_{j} \overline{w_j} (4\Re\inpro{v,x}) v_j  
= 4\inpro{w,v} \Re(\inpro{x,v}),
\end{align*}
and those for $(j,s)=(j_1,s_1)$ to be
\begin{align*}
\sum_{j,r,s} \sum_{r_1\atop j_2,r_2,s_2} &
\overline{i_r} i_s \overline{i_{r_1}} i_s \overline{i_{r_2}} i_{s_2} 
w_{jr} v_{jr_1} v_{j_2s_2} x_{j_2r_2} 
=\sum_{j,s,j_2} \overline{w_j} i_s \overline{v_j} i_s \overline{x_{j_2}} v_{j_2}
=\sum_{j} \overline{w_j} \Bigl(\sum_s i_s \overline{v_j} i_s \Bigr) \inpro{x,v} \cr
&=\sum_{j} \overline{w_j} (-2 v_j) \inpro{x,v}
=-2 \inpro{w,v}\inpro{x,v}.
\end{align*}
Hence
\begin{align*}
\inpro{w,D} |\inpro{v,x}|^2
&= 4\inpro{w,v} \Re\inpro{x,v}-2\inpro{w,v}\inpro{x,v} 
= 2\inpro{w,v}( 2\Re\inpro{x,v}-\inpro{x,v}) \cr
&= 2\inpro{w,v} \overline{\inpro{x,v}}
= 2\inpro{w,v} \inpro{v,x},
\end{align*}
and, similarly,
\begin{align*}
\inpro{D,w} |\inpro{v,x}|^2
& =  \sum_{j,r,s} \sum_{j_1,r_1,s_1\atop j_2,r_2,s_2}
\overline{i_r} i_s \overline{i_{r_1}} i_{s_1} \overline{i_{r_2}} i_{s_2} 
w_{js} v_{j_1r_1} v_{j_2s_2} 
(\gd_{jr,j_1s_1}x_{j_2r_2}+\gd_{jr,j_2r_2}x_{j_1s_1}) \cr
& = 4\Re(\inpro{\overline{w},\overline{v}})\inpro{x,v}-2\inpro{x,v}\inpro{w,v}
= 2\inpro{x,v}(2\Re\inpro{w,v}-\inpro{w,v}) \cr
&= 2\inpro{x,v}\overline{\inpro{w,v}}
= 2\inpro{x,v}\inpro{v,w}.
\end{align*}
Finally, by the first equation in (\ref{linlinderiv}), we have
\begin{align}
\label{DwwDcorrectcalc}
|\inpro{D,w}|^2|\inpro{v,\cdot}|^2
&= \inpro{w,D}(\inpro{D,w}|\inpro{v,\cdot}|^2)
= \inpro{w,D}(2\inpro{\cdot,v}\inpro{v,w}) \cr
&= 2m\inpro{w,v}\inpro{v,w}=2m|\inpro{v,w}|^2.
\end{align}
\end{proof}

We note the subtlety in the calculations above, e.g., the product rule holds if one factor is real-valued, 
but not if both are $\HH$-valued, and the differential operator $\inpro{w,D}$ does not commute with quaternion
scalars.


\begin{example}
\label{zonalinHomtt}
It follows from (\ref{planewavederiv}), that $|\inpro{w,D}|^2$ maps $\Hom_{\Fd}(t+1,t+1)$ to $\Hom_{\Fd}(t,t)$.
Using $\inpro{D,w}\norm{\cdot}^2=2\inpro{\cdot,w}$ and $\inpro{w,D}\norm{\cdot}^2=2\inpro{w,\cdot}$, 
we have
\begin{align*}
|\inpro{w,D}|^2\left(\norm{\cdot}^{2t}\right)
&= \inpro{w,D} \inpro{D,w} \left(\norm{\cdot}^{2t}\right)
= \inpro{w,D}\left(t\norm{\cdot}^{2(t-1)}(2\inpro{\cdot,w})\right)  \cr
&= 2t\norm{\cdot}^{2(t-1)} m\inpro{w,w}
+ t(t-1)\norm{\cdot}^{2(t-2)} (2\inpro{w,\cdot})(2\inpro{\cdot,w}) \cr
&= 2mt \norm{w}^2\norm{\cdot}^{2(t-1)} 
+ 4t(t-1) |\inpro{w,\cdot}|^2 \norm{\cdot}^{2(t-2)}, \qquad t\ge1,
\end{align*} 
so that $|\inpro{w,\cdot}|^2 \norm{\cdot}^{2(t-1)}\in\Hom_{\Fd}(t,t)$.
Continuing in this way, one obtains that
$|\inpro{v,\cdot}|^{2r}\norm{\cdot}^{2t-2r}\in\Hom_{\Fd}(t,t)$, $0\le r\le t$
(see \cite{MW20} for details).
\end{example}

\begin{theorem}
\label{Hdreproinproexist}
Let $\FF=\RR,\CC,\HH$.
There is a unique inner product on $\Hom_{\Fd}(t,t)$ for which 
$|\inpro{z,w}|^{2t}$ is the reproducing kernel. It is given by
\begin{equation}
\label{apolarFdefn}
\inpro{f,g}_\FF := {1\over b_{t,m}} \overline{f}(D) g, \qquad
b_{t,m}:=\prod_{j=1}^t 2j(2j+m-2).
\end{equation}
\end{theorem}

\begin{proof} We first show that $\dinpro{f,g}:=\overline{f}(D)g$
defines an inner product on $\Pol_{2t}(\Fd)$, and hence
on $\Hom_{\Fd}(t,t)$. Using the notation of (\ref{xjr-expansion}),
we may write 
$f\in\Pol_{2t}(\Fd)$ as
$$ f=\sum_{|\ga|=2t} f_\ga m_\ga, \qquad
m_\ga(x) := (x_{jr})_{1\le j\le d,1\le r\le m}^\ga, \quad f_\ga\in\FF. $$
Then
\begin{equation}
\label{monversionapolarFinpro}
\dinpro{f,g}
= \sum_{|\ga|=2t} \overline{f_\ga} m_\ga(D) \sum_{|\gb|=2t} g_\gb m_\gb
= \sum_{|\ga|=2t} \ga! \overline{f_\ga} g_\ga,
\end{equation}
which is clearly the (weighted) Euclidean inner product 
(and hence is an inner product).

By $t$ applications of (\ref{planewavederiv}) of Lemma \ref{planewavederivlemma}, 
we have
\begin{align*}
|\inpro{w,D}|^{2t} |\inpro{v,\cdot}|^{2t}
&= |\inpro{w,D}|^{2(t-2)} \left(|\inpro{w,D}|^2|\inpro{v,\cdot}|^{2t}\right) \cr
&= 2t(2t+m-2) |\inpro{v,w}|^2 \left( |\inpro{w,D}|^{2(t-2)} |\inpro{v,\cdot}|^{2t-2}\right) \cr
&= \cdots = b_{t,m} |\inpro{v,w}|^{2t}, 
\end{align*}
i.e., with $K_w(z)=|\inpro{z,w}|^{2t}$ and $f=|\inpro{v,\cdot}|^{2t}$,
$$ \inpro{K_w,f}_\FF = |\inpro{v,w}|^{2t} =f(w), \qquad\forall w\in\Fd, $$
so that $|\inpro{w,z}|^{2t}$ is the reproducing kernel for
$\inpro{\cdot,\cdot}_\FF = (1/b_{t,m})\dinpro{\cdot,\cdot}$.
\end{proof}

We will refer to (\ref{apolarFdefn}) as the {\bf apolar} inner product, since it coincides with the
apolar (Bombieri) inner product in the cases $\FF=\RR,\CC$,
i.e.,
$$ \inpro{f,g}_\RR= {1\over(2t)!}\overline{f}(D)g, \quad
\inpro{f,g}_\CC= {1 \over 2^{2t} t!^2 }\overline{f}(D)g, \quad
\inpro{f,g}_\HH= {1 \over 2^{2t}t!(t+1)!  }\overline{f}(D)g, $$
and see (\ref{BombieriInnerprod}) and (\ref{BombieriInnerprodCII}).
Again, we observe this is not the inner product given by integration on the (quaternionic) sphere.

\begin{example} We have
$$ \inpro{|q_j|^{2t},|q_j|^{2t}}_\HH = \inpro{K_{e_j},K_{e_j}}_\HH = |\inpro{e_j,e_j}|^{2t} =1, 
\qquad 1\le j\le d, $$
so the polynomials $p_j:q\mapsto q\mapsto |q_j|^{2t}$ in the
normalised tight frames of Theorem \ref{tightframeHomtttheorem} have unit norm, 
and hence (see \cite{W18} Exercise 2.4) span orthogonal one-dimensional subspaces of $\Hom_{\Hd}(t,t)$.
By evaluating $\inpro{|x_1|^{2t},|x_j|^{2t}}_\FF$ from (\ref{monversionapolarFinpro}), one gets the identity
$$ \sum_{|\ga|=t\atop\ga\in\ZZ_+^m} (2\ga)!{t\choose\ga}^2  = b_{t,m}
=\prod_{j=1}^t 2j(2j+m-2). $$
\end{example}

\section{Conclusion}

We have proved the quaternionic analogue of the Welch-Sidlenikov inequality
on the spacing of vectors/lines on unit sphere
(Theorem \ref{CSgeneralvarthm}),
and shown that equality in it corresponds to a cubature rule for the 
unitarily invariant polynomial space $\Hom_{\Hd}(t,t)$, which we call a
spherical $(t,t)$-design.
Consequences of this unified development for the real, complex and quaternionic cases include:

\begin{itemize}
\item A proof that projective spherical $t$-designs are precisely the spherical $(t,t)$-designs.
Since $(t,t)$-designs are cubature rules for the path-connected topological space $\SS$, 
it then follows from \cite{SZ84} that these exist for any given $t$. 
\item The variational characterisation gives a simple condition
for being a projective spherical $t$-design (Corollary \ref{regularschemecor}).
\item The variational characterisation allows for numerical constructions of
spherical $(t,t)$-designs (Examples \ref{numericalsearchI} and \ref{numericalsearchII}).
\item The polynomial space $\Hom_{\Fd}(t,t)$ plays a key role. In the real and complex cases
it is well understood with the bases given in Examples \ref{HomttRd} and \ref{HomttCd}
implying that
$$ \dim(\Hom_{\Rd}((t,t)) = {d+2t-1\choose 2t}, \qquad
\dim(\Hom_{\Cd}((t,t)) = {d+t-1\choose t}^2.  $$
We did not present a basis for $\Hom_{\Hd}(t,t)$, instead using a normalised tight frame
(Theorem \ref{tightframeHomtttheorem}). It can be shown (see \cite{MW20}) that
$$ \dim(\Hom_{\HH^d}(t,t)) = {1\over t+2d-1}  {t+2d-1\choose t}{t+2d-1\choose t+1}. $$
\end{itemize}

Directions for further investigation include results for the octonionic sphere, 
such as the possibility of a Welch-Sidlenikov inequality, and cubature rules
for unitarily invariant subspaces of $\Hom_{\Hd}(t,t)$.




\vskip2truecm

\bibliographystyle{alpha}
\bibliography{references}
\nocite{*}

\vfil\eject

\end{document}